\documentclass[sigconf, screen]{acmart}

\usepackage{tcolorbox}
\usepackage{graphicx}   

\usepackage{xcolor,listings}
\usepackage{textcomp}
\usepackage{color}
\usepackage{booktabs}
\usepackage{hyperref}   
\usepackage{thm-restate}
\usepackage{mdwlist}
\usepackage{mathtools}
\usepackage{xspace}
\usepackage{amsfonts}
\usepackage{verbatim}   
\usepackage{xcolor}
\usepackage{makecell}
\usepackage[mathscr]{euscript}
\usepackage{tikz}
\usetikzlibrary{calc}
\usetikzlibrary{arrows.meta}
\usepackage{enumitem}
\usepackage{relsize}
\usepackage{marginnote}
\usepackage{booktabs}
\usepackage{array}
\usepackage{tikz-qtree}
\usepackage{subcaption}
\usepackage{multirow}
\usepackage{makecell}
\usepackage{tabularx}
\usepackage{booktabs}

\usetikzlibrary{arrows.meta}
\usetikzlibrary{fit,shapes,trees,shapes.geometric}
\usetikzlibrary{patterns,decorations.pathreplacing,calc}
\usetikzlibrary{matrix, positioning, arrows}
\usetikzlibrary{chains,shapes.multipart}
\usetikzlibrary{shapes,calc}
\usetikzlibrary{automata}
\usepackage[linesnumbered,algoruled, lined, noend]{algorithm2e}

\newcommand{\introparagraph}[1]{\medskip \noindent {\bf  #1.}}  

\usepackage{aliascnt} 

\newtheorem{definition}{Definition}

\newtheorem{proposition}{Proposition}
\newtheorem{claim}{Claim}

\newtheorem{lemma}{Lemma}
\newtheorem{example}{Example}
\newtheorem{problem}{Problem}

\newtheorem{theorem}{Theorem}

\providecommand{\eat}[1]{}

\newcommand{\op}{\texttt{op}}
\newcommand{\sys}{\textsc{Rapidash}}
\newcommand{\sysg}{\textsc{Rapidash }}
\newcommand{\syso}{\textsc{Rapidash}$(\perp)$}
\newcommand{\syskd}{\textsc{Rapidash}$(\mathsf{kd})$}

\newcommand{\sysdisc}{\textsc{Rapidash}$(\mathsf{disc})$}
\newcommand{\facet}{\textsc{Facet}}
\newcommand{\vars}{\textsf{vars}}

\newcommand{\bX}{\mathbf{L}}
\newcommand{\bY}{\mathbf{U}}
\newcommand{\mL}{\mathcal{L}}
\newcommand{\TODO}[1]{\textcolor{red}{#1}}
\newcommand{\af}[1]{\textcolor{blue}{[Anna: #1]}}

\newcommand{\ra}[1]{\renewcommand{\arraystretch}{#1}}

\newcolumntype{Y}{>{\centering\arraybackslash}X}
\newcolumntype{L}{>{\leavevmode\color{magenta}}l}
\newcolumntype{C}{>{\leavevmode\color{magenta}}c|}
\newcolumntype{V}{>{\raggedleft\arraybackslash\leavevmode\color{magenta}}X}
\newcolumntype{R}{>{\raggedleft\arraybackslash}X}
\newcolumntype{N}{>{\raggedleft\arraybackslash\leavevmode\color{magenta}}p{8mm}}
\newcolumntype{M}{>{\centering\leavevmode\color{magenta}}p{0.9cm}}

\definecolor{light-gray}{gray}{0.95}

\def\mytokenshelp#1 #2\relax{\allowbreak\grayspace\tokenscolor{#1}\ifx\relax#2\else
 \mytokenshelp#2\relax\fi}
\newcommand\tokenscolor[1]{\colorbox{light-gray}{\textcolor{black}{%
  \ttfamily\mystrut\smash{\detokenize{#1}}}}}
\def\mystrut{\rule[\dimexpr-\dp\strutbox+\fboxsep]{0pt}{%
 \dimexpr\normalbaselineskip-2.25\fboxsep}}
\def\grayspace{\hspace{0pt minus \fboxsep}}

\definecolor{ForestGreen}{RGB}{34,139,34}

\usepackage{term}
\usepackage{local}

\newcommand{\mysimplenote}[1]{{#1}}

\newcommand{\fp}[1]{\textcolor{teal}{\mysimplenote{[fp]~#1}}}

\begin{document}

\title{\sys: Efficient Constraint Discovery via Rapid Verification}

\author{Zifan Liu}
\affiliation{%
    \institution{University of Wisconsin-Madison}
}
\email{zifan@cs.wisc.edu}

\author{Shaleen Deep}
\affiliation{%
    \institution{Microsoft}
}
\email{shaleen.deep@microsoft.com}

\author{Anna Fariha}
\affiliation{%
    \institution{University of Utah}
}
\email{afariha@cs.utah.edu}

\author{Fotis Psallidas}
\affiliation{%
    \institution{Microsoft}
}
\email{Fotis.Psallidas@microsoft.com}

\author{Ashish Tiwari}
\affiliation{%
    \institution{Microsoft}
}
\email{ashish.tiwari@microsoft.com}

\author{Avrilia Floratou}
\affiliation{%
    \institution{Microsoft}
}
\email{Avrilia.Floratou@microsoft.com}

\begin{abstract}
Denial Constraint (DC) is a well-established formalism that captures a wide range of integrity constraints commonly encountered, including candidate keys, functional dependencies, and ordering constraints, among others. Given their significance, there has been considerable research interest in achieving fast verification and discovery of exact DCs within the database community. Verification entails detecting whether a given DC holds true within a specific dataset, while discovery focuses on the automated mining of DCs satisfied on the dataset. Despite the significant advancements in the field, prior work exhibits notable limitations when confronted with large-scale datasets. The current state-of-the-art exact DC verification algorithm demonstrates a quadratic (worst-case) time complexity relative to the dataset's number of rows. In the context of DC discovery, existing methodologies rely on a two-step algorithm that commences with an expensive data structure-building phase, often requiring hours to complete even for datasets containing only a few million rows. Consequently, users are left without any insights into the DCs that hold on their dataset until this lengthy building phase concludes.

In this paper, we introduce \sys, a comprehensive framework for DC verification and discovery. Our work makes a dual contribution. First, we establish a connection between orthogonal range search and DC verification. We introduce a novel exact DC verification algorithm that demonstrates near-linear time complexity, representing a theoretical improvement over prior work. Second, we propose an anytime DC discovery algorithm that leverages our novel verification algorithm to gradually provide DCs to users, eliminating the need for the time-intensive building phase observed in prior work. To validate the effectiveness of our algorithms, we conduct extensive evaluations on four large-scale production datasets. Our results reveal that our DC verification algorithm achieves up to $40 \times$ faster performance compared to state-of-the-art approaches. Furthermore, we demonstrate the superiority of our DC discovery algorithm by showcasing its ability to produce constraints within the initial $10$ minutes of execution, while prior methods fail to generate any output within the first $48$ hours of execution.

\end{abstract}

\maketitle

\section{Introduction}  \label{sec:intro}

Integrity constraints play a pivotal role in a wide range of data analysis tasks such as data exploration~\cite{abedjan2016detecting, Fariha0RGM21}, data cleaning and repair~\cite{rekatsinas2017holoclean, giannakopoulou2020cleaning}, data synthesis~\cite{gekamino}, and query optimization~\cite{kossmann2022data}. By enforcing integrity constraints, organizations can ensure the reliability, consistency, and accuracy of their data, enabling them to make informed decisions, derive meaningful insights, and extract maximum value from their datasets. One class of constraints that is of particular interest is known as \textit{Denial Constraints} (DCs)~\cite{chu2013discovering}. DCs are appealing since they are expressive enough to capture many integrity constraints that are useful in practice such as functional dependencies, ordering constraints, and unique column combinations among others. 


\begin{example} \label{ex:1}
    Table~\ref{table:tax} shows a sample of a tax dataset  that contains information about tax rates for people in different US states. Several rules are true about this dataset: (1) \texttt{SSN} column is a candidate key, (2) \texttt{Zip}$\rightarrow$\texttt{State} is a functional dependency, and (3) for all people in the same state, the person with higher salary has a higher tax rate. Each of these rules can be expressed as DCs as we will see in Section~\ref{sec:prelim}.
\end{example}

\begin{table}[t]
    \small
    \caption{\texttt{Tax}
rates for people in different states in the USA.} \label{table:tax}
    \vspace{-3mm}
     \begin{tabularx}{\linewidth}{ @{\extracolsep{\fill}}l l l r r l} 
     \toprule
     & \texttt{SSN} & \texttt{Zip} & \texttt{Salary} & \texttt{FedTaxRate} & \texttt{State} \\
     \midrule
    $t_1$ & 100 & 10108 & 3000 & 20\% & New York \\
    $t_2$ & 101 & 53703 & 5000 & 15\% & Wisconsin \\
    $t_3$ & 102 & 53703 & 6000 & 20\% & Wisconsin \\
    $t_4$ & 103 & 53703 & 4000 & 10\% & Wisconsin \\
     \bottomrule
     \vspace{-5mm}
     \end{tabularx}
     \end{table}

\looseness-1
Our work emphasizes on two important scenarios: \textbf{exact} DC \textbf{verification} and \textbf{discovery}. DC verification involves detecting whether a given DC is satisfied on a specific dataset and is particularly valuable during data exploration, where analysts aim to ascertain the presence or absence of specific patterns within the dataset. Additionally, it serves as a valuable tool in assessing dataset quality, enabling the identification of noisy or inconsistent data instances~\cite{Fariha0RGM21}. DC discovery involves the automatic discovery of exact DCs from a given dataset, which holds significant appeal due to the inherent challenges associated with the manual identification of DCs. The manual approach not only requires expertise and significant time investment but also suffers from a higher likelihood of errors, given the intricate and ever-evolving nature of datasets.

\looseness-1 In recent years, substantial advancements happened in the field of exact DC verification and discovery~\cite{pena2019discovery, pena2020efficient, pena2021fast, pena2022fast}. However, our practical experience in applying some of these approaches to real-world production datasets has unveiled noteworthy limitations of existing methods (refer to Section~\ref{sec:evaluation} for comprehensive details).
First, in the context of DC verification, the best-known algorithm, \facet~\cite{pena2021fast}, has a worst-case time and space complexity  $\Omega(|\R{R}|^2)$ on a given relation $\R{R}$ with cardinality $|\R{R}|$ (number of rows). In this work, we make a connection between the problem of DC verification and \emph {orthogonal range search}~\cite{bentley1979data, bentley1980decomposable}, a celebrated line of work in computational geometry, which studies the problem of determining which $k$-dimensional objects in a set intersect with a given query object. We show that it is possible to design a near-optimal algorithm for verifying a given DC over a specific dataset by leveraging techniques employed for orthogonal range search. Our proposed algorithm has a time complexity of $O(|\R{R}| \log^{f(\varphi)} |\R{R}|)$, where $f(\varphi)$ is a parameter that is dependent only on the characteristics of the DC $\varphi$ and not on the input dataset $\R{R}$. This represents a theoretical improvement over prior work and translates into an order of magnitude better performance in practice. 

In the context of exact DC discovery, prior work~\cite{chu2013discovering, bleifuss2017efficient, pena2019discovery, pena2020efficient, pena2021fast} follows a two-step process: (1)~building an intermediate data structure called \textit{evidence set} from the input, which is the most computationally demanding aspect of the DC discovery process~\cite{pena2022fast}, and (2)~mining the DCs from the evidence set using various set-covering algorithms, which could also be costly depending on the number of DCs and the size of the evidence set.
Our experience of applying this two-phase approach has unveiled that the time required to construct the evidence set is often prohibitive. Even for medium-sized datasets (e.g., $5$ million rows and $30$ columns), it can take up to several hours just to construct the evidence set. 
Recent work~\cite{pena2022fast} has shown that parallelization can reduce the time taken for evidence-set construction. However, parallelization improves performance by a constant factor, and, thus, is not a substitute for better worst-case complexity. Moreover, it does not necessarily lead to more scalable performance as the dataset size increases since there is a limit in the degree of parallelism.


After talking to various customers to better understand their requirements and expectations regarding constraint discovery\footnote{We are omitting more details due to double-anonymization considerations.}, we came to the conclusion that this two-phase approach has a fundamental limitation that can lead to poor user experience.
In particular, the approach is ``all or none'', i.e., to produce any DC, it needs to complete the full evidence set construction (which is time-consuming), at the end of which it reports all DCs. However, in practical scenarios, users often have distinct preferences and requirements. Firstly, they prioritize receiving confirmed DCs promptly, starting from simpler constraints and gradually progressing towards more complex ones ($R_{1}$). Secondly, users value the flexibility of terminating the discovery process prematurely if they are satisfied with the set of DCs already mined at any specific point in time ($R_{2}$). This perspective underscores the necessity for a more flexible DC-discovery process to enhance the overall user experience and motivates the need for designing a new solution. 


In this work, we propose a novel \textit{anytime}~\cite{zilberstein1996using} DC discovery algorithm that allows for progressive constraint discovery and early  termination, and, thus, satisfies requirements $R_{1}$ and $R_{2}$. At a high-level, our algorithm performs a lattice-based traversal of the space of DCs and invokes our novel DC verification algorithm to confirm whether a given constraint holds. Unlike prior work, our algorithm does not have a blocking building phase and bypasses the evidence-set-construction-based paradigm.

\smallskip
\introparagraph{Our contributions} Our key contribution is a general framework, \sys, that relies on a novel approach for exact DC verification and discovery leveraging the connection to orthogonal range search. Specifically, we make the following contributions: 

\introparagraph{1. \emph{A novel DC verification algorithm}} We present a near-optimal algorithm for verifying a given DC on a dataset $\R{R}$. We prove that our proposed algorithm can achieve a near-linear time and space complexity wrt.\ dataset size. This represents a significant improvement over the best-known verification algorithm~\cite{pena2021fast}, which has a worst-case quadratic complexity (both time and space). Further, we show that in certain scenarios, our algorithm can run in only linear space while still achieving provably sub-quadratic running time.

\introparagraph{2. \emph{Efficient DC discovery}} We introduce the problem of \textit{anytime} DC discovery and propose a lattice-based algorithm that relies on our novel DC verification algorithm to provide better performance than prior work, which relies on evidence sets. 

\looseness-1
\introparagraph{3. \emph{Experimental evaluation}} We conduct an extensive empirical evaluation over four production datasets that are an order of magnitude larger than those used in prior work. We show that \sysg achieves up to $40\times$ speedup over the state of the art~\cite{pena2021fast} for exact DC verification. For DC discovery, our anytime algorithm can produce all single-column DCs (e.g., \eat{whether a column is sorted or not, }single-column candidate keys, column with identical values etc.) within the first $10$ minutes, while prior work~\cite{bleifuss2017efficient, pena2019discovery} fails to produce any output within the first $48$ hours. We also show that \sysg\ scales better 
than prior work. 

\eat{The rest of the paper is organized as follows: Section~\ref{sec:prelim} provides background on denial constraints and presents our problem statement. Section~\ref{sec:limitation} discusses limitations of prior work and Section~\ref{sec:verification} presents the design of our proposed DC verification algorithm along with various optimizations to handle different classes of constraints. Section~\ref{sec:discovery} presents the design of our DC discovery algorithm. Section~\ref{sec:evaluation} presents our experimental findings and Section~\ref{sec:related} discusses related work. Finally, Section~\ref{sec:conclusions} concludes the paper and discusses opportunities for future work.}

\eat{In this paper, we propose a novel DC discovery algorithm that addresses all of the above limitations. The key insight that we exploit is that a fast DC verification algorithm used in conjunction with a lattice-based approach for DC discovery is sufficient to satisfy all our requirements. Lattice-based approaches have been used in prior works for dependency discovery (such as functional dependencies). However, lattice-based methods have mostly been dismissed for DC discovery since building lattices over column combinations might be prohibitive. However, all prior work on DC discovery implicitly assumes that the users are interested in obtaining all DCs that hold for a given dataset. Based on our analysis and discussion with customers as well as engineers in the data quality domain, we find that this is not the case. Most customers prefer rules that are interpretable and easy to understand. Prior work has already captured this property by remarking that \textit{succintness} of the rules is a key property for identifying which constraints are ranked higher. Concretely, this translates to generating constraints that span at most a few columns. This property in turn helps us ensure that only a fixed number of levels in the lattice need to be explored. We make the following contributions.}

\section{Background}  \label{sec:prelim}
In this section, we provide background on terminology and notations that will be used throughout the paper. We also define the two problems that we are tackling (DC verification and discovery).

\introparagraph{Relations} Let \R{R}\ be the input relation and \vars$(\R{R})$ denote the finite set of attributes (i.e. the columns). We use $|\R{R}|$ to denote the cardinality of the relation. We will use \A{A}, \A{B}\ to denote attributes, \tup{s}\ and \tup{t}\ to denote \emph{tuples}, and \proj{t}{\A{A}}\ denotes the value of an attribute \A{A}\ in a tuple \tup{t}. Throughout the paper, we assume bag semantics where the relation can have the same tuple present multiple times. \eat{We will use ${w} \subseteq \vars(\R{R})$ to denote a subset of variables.}

\introparagraph{Denial Constraints (DCs)} DCs express predicate conjunctions to determine conflicting combinations of column values. They generalize other integrity constraints, including unique column combinations, functional dependencies, and order dependencies. We define a predicate $p$ as the expression $\proj{s}{\A{A}}\ \op\ \proj{t}{\A{B}}$ where $\tup{s}, \tup{t} \in \R{R}$, $\op \in \{=,\neq, \geq, >, \leq, <\}$ and $\A{A}, \A{B} \in \vars(\R{R})$. We will refer to $\neq$ as disequality and $\geq, >, \leq, <$ as inequalities. All operators except equality will be collectively referred to as non-equality operators. A DC $\varphi$ is a conjunction of predicates of the following form: 

$$ \forall \tup{s}, \tup{t} \in R, s \neq t: \quad \neg (p_1 \land \dots \land p_m) $$

A tuple pair $(\tup{s}, \tup{t})$ is said to be a violation if all predicates in $\varphi$ evaluate to true. We will say that a $\varphi$ holds on $\R{R}$ if there are no violations, i.e., the DC is \emph{exact}. An exact DC is said to be minimal if no proper subset of its predicates forms another exact DC.

A predicate is said to be \emph{homogeneous} if it is of the form $\proj{s}{\A{A}}\ \op\ \proj{t}{\A{A}}$ or $\proj{s}{\A{A}}\ \op\ \proj{s}{\A{B}}$, i.e. it is either defined over a single column $A$ or it is defined over a single tuple but two different columns, and \emph{heterogeneous} if it is of the form $\proj{s}{\A{A}}\ \op\ \proj{t}{\A{B}}$. We will refer to $\proj{s}{\A{A}}\ \op\ \proj{t}{\A{A}}$ as row-level homogeneous predicate since such a predicate is comparing across two rows and $\proj{s}{\A{A}}\ \op\ \proj{s}{\A{B}}$ as column-level homogeneous predicate since it compares two columns of the same row. 

Since most DCs of interest contain only row-level homogeneous predicates (such as ordering dependencies~\cite{grya2012fundamentaab}, functional dependencies, candidate keys, etc.), for simplicity, we will use the term homogeneous DC to refer a DC that contains only row-level homogeneous predicates. We will use the term mixed homogeneous DC to refer to DCs that contain both row and column-level homogeneous DC. A heterogeneous DC can contain all types of predicates. Without loss of generality, we will assume that each column of \R{R}\ participates in at most one predicate of a homogeneous DC. We will use $\vars_{\op}(\varphi)$ to denote the set of columns in a homogeneous DC that appear in some predicate with the operator as $\op$.

\begin{example}
    Continuing from \autoref{ex:1}, each constraint can be expressed using a DC as follows: (1) $\varphi_1: \neg(\proj{s}{\mathtt{SSN}} = \proj{t}{\mathtt{SSN}})$; (2) $\varphi_2: \neg(\proj{s}{\mathtt{Zip}} = \proj{t}{\mathtt{Zip}} \land \proj{s}{\mathtt{State}} \neq \proj{t}{\mathtt{State}})$; (3) $\varphi_3: \neg(\proj{s}{\mathtt{State}} = \proj{t}{\mathtt{State}} \land \proj{s}{\mathtt{Salary}} <  \proj{t}{\mathtt{Salary}} \land \proj{s}{\mathtt{FedTaxRate}} >  \proj{t}{\mathtt{FedTaxRate}})$. The universal quantification is left implicit. Let us fix our attention to $\varphi_3$. Note that $\vars_{=}(\varphi_3) = \{\mathtt{State}\}, \vars_{<}(\varphi_3) = \{\mathtt{Salary}\}$, and $\vars_{>}(\varphi_3) = \{\mathtt{FedTaxRate}\}$.

    \smallskip
    \noindent Observe that all the DCs above are homogeneous (i.e. contain only row-level homogeneous predicates). An example of a heterogeneous DC is $\varphi_4 : \neg(\proj{s}{\mathtt{Salary}} < \proj{t}{\mathtt{FedTaxRate}})$.
    All the DCs hold on the relation $\mathtt{Tax}$ defined in Table~\ref{table:tax} and are minimal exact DCs. \qed
\end{example}

\eat{\fp{numerical attributes can also be categorical. maybe change to continuous?.}}

\introparagraph{Predicate Space} The space of DCs is governed by the \emph{predicate space}, the set of all predicates that are allowed on $\R{R}$. As noted in~\cite{pena2019discovery, pena2020efficient}, a predicate is meaningful when a proper comparison operator is applied to a pair of comparable attributes. Specifically, all the six operators can be used on numerical attributes (i.e. they are continuous), e.g., age and salary, but only $=$ and $\neq$ 
can be used on categorical attributes such as name and address. Two attributes are said to be comparable if: $(i)$ they have the same type; $(ii)$ the active domain overlap is at least $30\%$~\cite{pena2019discovery, pena2020efficient}. For example, in~\autoref{ex:1}, column \texttt{Salary} and \texttt{State} are not comparable since they have different type, and \texttt{SSN} and \texttt{Zip} are not comparable since the values do not have any overlap.

\subsection{Problem Statement}

We use the term {\em{DC verification}} for the process of determining whether a DC holds on a relation \R{R}\ and {\em{DC discovery}}\footnote{We will use the term discovery and mining interchangeably.} to refer to the process of finding (some or all) exact, minimal DCs over \R{R}. In this paper, we focus on the following two problems.

\eat{\fp{what does ``scalable'' mean in Problem 2?}}

\begin{problem}
    Given a relation \R{R}\ and a DC $\varphi$, determine whether $\varphi$ holds on \R{R}.
\end{problem}

\begin{problem}
    Given a relation \R{R}, design an efficient, anytime DC discovery algorithm.
\end{problem}
An anytime algorithm is required to produce an increasing number of exact DCs as time progresses in a way that we have some exact DCs even if the algorithm is interrupted before it terminates. 

\introparagraph{Computational Model} We focus on evaluation in the main-memory setting. We assume the RAM~\cite{hopcroft2001introduction} model of computation where tuple values and integers take $O(1)$ space and arithmetic operations on integers, as well as
memory lookups, are $O(1)$ operations. Further, we assume perfect hashing for our hash tables where insertions and deletions can be reflected in $O(1)$ time and a hash table takes space linear in the number of entries it stores. Throughout the paper, we will consider the data complexity of the problems where the DC size is assumed to be a constant.
\section{Limitations of existing Solutions} \label{sec:limitation}

We now discuss the limitations of the existing solutions for exact DC verification and discovery. In Section~\ref{sec:evaluation}, we experimentally demonstrate some of these limitations using real-world datasets.

\smallskip
\introparagraph{DC Verification} We begin by giving a brief description of the key ideas underlying \textsc{Facet}, the state-of-the-art system for DC verification. Let \textsf{tids} denote a set of tuple identifiers. All tuples in relation $\R{R}$ can be represented as $\mathsf{tids}_{\R{R}} = \{\tup{t}_1, \dots, \tup{t}_{|\R{R}|}\}$. An ordered pair $(\mathsf{tids}_1, \mathsf{tids}_2)$ represents all tuples pairs $(\tup{s},\tup{t})$ such that $\tup{s} \in \mathsf{tids}_1, \tup{t} \in \mathsf{tids}_2, \tup{s} \neq \tup{t}$. \textsc{Facet} processes one predicate of the DC at a time, taking a set of ordered pairs $(\mathsf{tids}_1, \mathsf{tids}_2)$ as input and generating another set of ordered pairs $(\mathsf{tids}'_1, \mathsf{tids}'_2)$ that represent tuple pairs that satisfy the predicate as the output. This process is known as \emph{refinement} and \facet\ refines each predicate using specialized algorithms for each operator. The output of a refinement is consumed as the input for refining the next predicate. At the end of processing all the predicates, we get all tuples pairs that satisfy all the predicates and thus, represent the violations.

\eat{\fp{I may be a bit lost with the following example, but $\varphi_3$ does not contain a negation (in contrast to the one defined earlier). hence, the last sentence on ``violation'' is not really violations but rather what satisfies $\varphi_3$. }}

\eat{\fp{also, it was not clear that $\mathtt{FedTaxRate} = {\color{red} \mathbf{22}}$ was trying to set the original example to a different value. (btw, why not having 22 in the initial example, and setting it here to something else so that the original example leads to a non-empty output for $\varphi_3$? if this is because through the empty set the predicates holds on the whole dataset, then maybe say it here ``...thus, the output is the empty set. Hence, $\varphi_3$ holds on the whole dataset.''}}

\begin{example} \label{ex:facet}
    Consider the DC $\varphi_3: \neg (\proj{s}{\mathtt{State}} = \proj{t}{\mathtt{State}} \land \proj{s}{\mathtt{Salary}} <  \proj{t}{\mathtt{Salary}} \land \proj{s}{\mathtt{FedTaxRate}} >  \proj{t}{\mathtt{FedTaxRate}})$. The refinement of predicate $p_1: \proj{s}{\mathtt{State}} = \proj{t}{\mathtt{State}}$ produces the set $\{(\{t_2, t_3, t_4\}$ $,\{t_2, t_3, t_4\})\}$ with a single ordered pair. This ordered pair represents the set of tuple pairs: $(t_2, t_3),$ $(t_2, t_4),$ $(t_3, t_2), (t_3, t_4),$ $(t_4, t_2), (t_4, t_3)$ since each of them satisfy $p_1$. Next, this singleton set is provided as input to predicate $p_2: \proj{s}{\mathtt{Salary}} <  \proj{t}{\mathtt{Salary}}$ which produces a new set $\{(\{t_4\}, \{t_2, t_3\}), (\{t_2\}, \{t_3\})\}$ since the \texttt{Salary} for $t_4$ is smaller than both $t_2$ and $t_3$ but \texttt{Salary} for $t_2$ is smaller only than $t_3$. Finally, we process predicate $p_3: (\proj{s}{\mathtt{FedTaxRate}} > \proj{t}{\mathtt{FedTaxRate}})$. However, note that none of the tuple pairs $(t_4, t_2), (t_4, t_3), (t_2, t_3)$ satisfy the predicate and thus, the output is the empty set. Hence, $\phi_3$ holds on the whole dataset $\mathtt{Tax}$. Let us modify $\mathtt{Tax}$ by setting $t_4.\mathtt{FedTaxRate}$ to  ${\color{black} \mathbf{22}}$ and call it $\mathtt{Tax'}$. Then the output of the refinement by predicate $p_3$ would be $\{(\{t_4\}, \{t_2\}), (\{t_4\}, \{t_3\})\}$ which represents the two violations of $\varphi_3$ on $\mathtt{Tax'}$. \qed
\end{example}

\noindent \facet\ contains algorithms that are custom-designed for the different predicate structures. We now highlight the three key sources of inefficiency in \textsc{Facet}.

\begin{enumerate}[wide, labelwidth=!, labelindent=0pt]
    \item \textbf{Complexity of \textsf{IEJoin}.} \textsc{Facet} and \textsc{Hydra} both use \textsf{IEJoin} \cite{khayyat2015lightning} as the algorithm for processing inequalities. The algorithm is designed to process two inequalities at a time and thus operates on two sets of tuple pairs simultaneously (akin to two relations). The running time complexity of \textsf{IEJoin} is $O(|R| \cdot |S|)$ for processing inequality joins between two relations $R$ and $S$ (although its space complexity is only $O(|R| + |S|)$). As noted in~\cite{pena2021fast}, \textsf{IEJoin} is severely under-performing for predicates of low selectivity.

    \item \textbf{Complexity of \textsf{Hash-Sort-Merge}.} Since \textsf{IEJoin} is designed for at least two predicates with inequality, \textsc{Facet} proposed two novel optimizations to process DCs with a single inequality predicate: \texttt{\bfseries Hash-Sort-Merge (HSM)} and \texttt{\bfseries Binning-Hash-Sort-Merge (BHSM)}. However, it can be shown that both \texttt{HSM} and \texttt{BHSM} still require a quadratic amount of running time and space in the worst-case. Similarly, processing of predicates containing disequality also requires quadratic time and space in the worst-case.

    \item Since \facet\ processes one predicate at a time, it needs to make at least one full pass over the dataset. As we will see later, this is not always necessary.
\end{enumerate}

\smallskip

\begin{sloppypar}
\introparagraph{DC Discovery} As mentioned in Section~\ref{sec:intro}, the first (and the most expensive) step performed by existing DC discovery algorithms is the computation of the evidence set. Given a predicate space $P$ and a pair of tuples $(s, t)$, the evidence $e(s, t)\subseteq P$ is the subset of predicates satisfied by the tuple pair. The evidence set is the set of evidences for all tuple pairs in the dataset. For example, in Table~\ref{table:tax} assuming the predicate space $P = \{ p_1: s.\mathtt{SSN} \neq t.\mathtt{SSN}, p_2: s.\mathtt{Zip} \neq t.\mathtt{Zip}, p_3: s.\mathtt{Zip} = t.\mathtt{Zip}, p_4: s.\mathtt{FedTaxRate} \neq t.\mathtt{FedTaxRate}, p_5: s.\mathtt{FedTaxRate} = t.\mathtt{FedTaxRate}, p_6: s.\mathtt{FedTaxRate} > t.\mathtt{FedTaxRate}, p_7:  s.\mathtt{FedTaxRate} < t.\mathtt{FedTaxRate} \}$, the evidences $e$ for all the tuple pairs $(t_i, t_j)$ are as follows (we show the cases where $i < j$, and the rest can be implied by symmetricity): 
\end{sloppypar}
\begin{align*}
    e(t_1, t_2) &= \{p_1, p_2, p_4, p_6\}, \quad e(t_1, t_3) = \{p_1, p_2, p_5\} \\
    e(t_1, t_4) &= \{p_1, p_2, p_4, p_6\}, \quad e(t_2, t_3) = \{p_1, p_3, p_4, p_7\} \\
    e(t_2, t_4) &= \{p_1, p_3, p_4, p_7\}, \quad e(t_3, t_4) = \{p_1, p_3, p_4, p_6\}
\end{align*}

\looseness-1
The evidence set will contain $4$ evidences since $e(t_1, t_2)$ and $e(t_1, t_4)$ are identical (and so are $e(t_2, t_3)$ and $e(t_2, t_4)$). As discussed before, evidence set construction is a blocking step since discovery cannot start until the computation has been completed. The time complexity of the construction process is (possibly) super-linear dependency on $|\R{R}|$ depending on the characteristics of the tuples and columns in the input. Our experiments in Section~\ref{sec:evaluation} demonstrate super-linear (closer to $|\R{R}|^{3/2}$) complexity in practice on our datasets. In terms of space, in the worst case, the size of the evidence set could be as large as $|\R{R}|^2$, which is undesirable. 

These drawbacks motivate the necessity for designing a new algorithm that has the \textit{anytime} property. The reader may wonder whether it is possible to adjust evidence set construction to enable an anytime DC discovery algorithm that starts emitting simpler constraints progressing towards more complex ones over time. Intuitively, such an algorithm would be possible if we could create an evidence set catered to DC constraints consisting of one predicate only, discover the ones that are satisfied and return them to the user, increment the existing evidence set to cover constraints with two predicates and repeat the process until the full space of constraints has been explored or the user terminates the process. However, as shown in the example above, the evidence set construction relies on the predicate space and not the DC constraint space. As a result, the evidence set used to mine constraints with one predicate is exactly the same as the one used to mine constraints with two predicates. Thus, incremental construction of the evidence set (and evidence set-based anytime DC discovery) is unlikely.

\section{\sysg Verification} \label{sec:verification}

In this section, we describe the \sysg verification algorithm. Our algorithm builds appropriate data structures to store the input data (leveraging existing work on orthogonal range search), and issues appropriate queries to find violations of a given DC. 



\eat{Without loss of generality, we will assume that all predicates of the DC contain only equalities and inequalities but no disequality. Indeed, any DC containing disequalities can be expanded into a set of equivalent DCs that contain only equalities and inequalities. We will remove the homogeneous assumption on $\varphi$ later. First, we present some background on orthogonal range search that will be useful.}

\subsection{Orthogonal Range Search}

\eat{\fp{A also needs to be $ \in \mathbb{N}^k$}}

In this section, we present some background on orthogonal range search. Given a totally ordered domain $\mathbb{N}$, let $A \subseteq \mathbb{N}^k$, for some $k \geq 1$, of size $N$. Let $\bX = (\ell_1, \dots, \ell_k)$ and $\bY = (u_1, \dots, u_k)$ be such that $\bX, \bY \in \mathbb{N}^k$ and $\ell_i \leq u_i$ for all $i \in [k]$. 

\begin{definition} \label{def}
An orthogonal range search query is denoted by $(\bX, \bY)$, and its evaluation over $A$ consists of enumerating the set 

$$Q(A) = \{ a \in A \mid \bigwedge \limits_{i \in k} \ell_i\ \op_1\ a_i\ \op_2\ u_i \}$$

where $\op_1, \op_2 \in \{<, \leq\}$.
\end{definition}

In other words, $\bX$ and $\bY$ form an axis-aligned hypercube in $k$ dimensions, and $Q(A)$ reports all points in $A$ that lie on/within the hypercube. The Boolean version of the orthogonal range search problem consists of determining if $Q(A)$ is empty or not.

\begin{example} \sloppy
    Consider the Table~\texttt{Tax} from \autoref{ex:1}. Let $A$ be the set of two-dimensional points obtained by projecting \texttt{Tax} on ${\mathtt{Salary}}$ and $\mathtt{FedTaxRate}$. Let $\bX = (3500, 5)$ and $\bY = (4500, 22)$. Then, the orthogonal range query $(\bX, \bY)$ is asking for all points such that the \texttt{Salary} is between $3500$ and $4500$, and the \texttt{FedTaxRate} is between $5$ and $22$. In Table~{\texttt{Tax}}, only $\tup{t}_4$ satisfies the criteria (its values of \texttt{Salary} and \texttt{FedTaxRate} are $4000$ and $10$ respectively). Thus, the result of the orthogonal range search query $(\bX, \bY)$ is $\{(4000, 10)\}$. \qed
\end{example}

In the presentation of the algorithms, we will assume that the range search data structure is built over $k$ dimensions and has two methods in its API: 

\begin{enumerate}[leftmargin=*]
    \item \textsf{booleanRangeSearch}$(\bX, \bY)$: returns a Boolean value if there is a point that lies in the axis-aligned hypercube formed by $\bX$ and $\bY$. The operators $\op_1$ and $\op_2$ used in~\autoref{def} will be clear from the context in which the function is called.
    \item \textsf{insert}$(t)$: inserts a $k$-dimensional tuple $t$ into the data structure.
\end{enumerate}

The two most celebrated data structures for orthogonal range search that are widely used in practice are range trees~\cite{bentley1979data} and $k$-d trees~\cite{bentley1979data}. We will review their complexity and trade-offs when analyzing the complexity of our DC verification algorithm.

\eat{Seminal work by Bentley et al.~\cite{bentley1979data} presented two algorithms that can answer orthogonal range search queries efficiently by trading off data structure construction time and space with query answering time. These algorithms leverage two different tree structures (range trees and k-d trees) and their time/space characteristics are presented below.

\begin{theorem}[\cite{bentley1979data}] \label{thm:range}
    Given $N$ k-dimensional points, there exists a data structure, known as a \textbf{range tree}, that can be built in pre-processing time $P(N) = O(N \log^{k-1} N)$, using $S(N) = O(N \log^{k-1} N)$ space, and can answer any orthogonal range search query\footnote{Throughout the paper, we use $\log^k N$ to mean $(\log N)^k$ and not iterated logarithms.} $Q(A)$ in time $T(N) = O(\log^{k-1} N + |Q(A)|)$.
\end{theorem}

\begin{theorem}[\cite{bentley1979data}] \label{thm:kd}
    Given $N$ k-dimensional points, there exists a data structure, known as a \textbf{k-d tree} that can be built in pre-processing time $P(N) = O(N \log N)$, using $S(N) = O(N)$ space, and can answer any orthogonal range search query $Q(A)$ in time $T(N) = O(N^{1-\frac{1}{k}} + |Q(A)|)$.
\end{theorem}

Note that any algorithm that evaluates the orthogonal range search query can solve the Boolean version of the problem by simply stopping the algorithm as soon as we identify that the output is non-empty (i.e. the algorithm is stopped as soon as we know that $|Q(A)| \geq 1$).

Our \sysg algorithm relies on orthogonal range search and can leverage both range and k-d trees. We present experimental results with both variants in Section~\ref{sec:evaluation}. Before presenting the details of the algorithm, we also introduce dynamic range which is used by \sysg for incremental DC verification.

{\color{red} AF: This section needs more work}
\smallskip
\introparagraph{Dynamic range search} Using standard \emph{static-to-dynamic} transformation techniques~\cite{bentley1980decomposable, overmars1983design}, it is well-known that the static data structure for range searching can be transformed into dynamic data structures where the points can be incrementally inserted and deleted. 

\begin{proposition} \label{prop:incremental}
    Given a static range searching data structure with pre-processing time $P(n)$ and query time $T(n)$ on an input of size $n$, there exists a data structure where a point can be inserted in time $O(P(n) \log n / n)$ and the query time is $O(T(n) \log n)$.
\end{proposition}

We note that the data structures generated by the transformation are practical and are shipped in standard programming language libraries~\cite{haskellkd}.}

\begin{algorithm}[!t]
    \small
	\SetCommentSty{textsf}
	\DontPrintSemicolon 
	\SetKwInOut{Input}{Input}
	\SetKwInOut{Output}{Output}
	\SetKwRepeat{Do}{do}{while}
	\SetKwFunction{rangeSearch}{\textsc{booleanRangeSearch}}
	\SetKwFunction{ins}{\textsc{insert}}
        \SetKwFunction{lowerrange}{\textsc{SearchRange}}
        \SetKwFunction{invertrange}{\textsc{InvertRange}}
	\SetKwProg{myalg}{procedure}{}{}
	\SetKwData{tree}{\textsc{OrthogonalRangeSearch}}
	\SetKwData{op}{$\mathsf{op}$}
        \SetKwData{nonEqCount}{$\mathsf{k}$}
	\SetKwData{col}{$\mathsf{col}$}
	\Input{Relation $\R{R}$, Homogeneous DC $\varphi$}
        \Output{True/False}
	\SetKwProg{myproc}{\textsc{procedure}}{}{}
	\SetKwData{ret}{\textbf{return}}
        $\nonEqCount \leftarrow |\vars(\varphi) \setminus \vars_{=}(\varphi)|$ \;\label{line:non-eq-count}
        $H \leftarrow $ empty hash table \;
	\ForEach{$t \in R$}{
            $v \leftarrow \pi_{\vars_{=}(\varphi)} (t)$ \;\label{line:4}
            \If{$v \not\in H$}{ \label{line:init-begin}
                \If{$\nonEqCount \neq 0$}{
                    $H[v] \leftarrow \text{new } \tree()$ \label{line:5}
                } \Else {
                    $H[v] \leftarrow 0$ \label{line:init-end}
                }
            }
            \If{$\nonEqCount \neq 0$}{ \label{line:contains-inequality}
                    $\bX, \bY \leftarrow \lowerrange{t}$ \; \label{line:check-begin}
                    $\bX', \bY' \leftarrow \invertrange(\bX, \bY)$ \;
                    \tcc{$\op_1$ and $\op_2$ for \rangeSearch\ are chosen based on $\op$ in the inequality predicates}
            	\If{$H[v].\rangeSearch(\bX, \bY) \lor         
                    H[v].\rangeSearch(\bX', \bY')$ \label{line:search}}{
                        \ret \textbf{false} \label{line:return-false-inequality}
                   }
                   $H[v].\ins(\pi_{\vars(\varphi) \setminus \vars_{=}(\varphi)}(t))$ \label{line:insert}
            } \Else {
                $H[v] \leftarrow H[v] + 1$ \label{increment}\;
                \If{$H[v] > 1$}{
                    \ret \textbf{false} \label{line:end}
                }
            }
            
        }
        \ret \textbf{true} \; \label{line:return-true}
        \myproc{\lowerrange{t}}{
            $\bX \leftarrow (-\infty, \dots, -\infty), \bY \leftarrow (\infty, \dots, \infty)$ \tcc*{$\bX$ and $\bY$ are indexed by the non-equality predicates $p_i$. Both are of size $\nonEqCount$}
            \ForEach{predicate $p_i \in \text{non-equality predicates in $\varphi$}$}{
                \If{$p_i.\op$ is $<$ or $\leq$}{
                    $\bY_i \leftarrow \min \{\bY_i, \pi_{p_i.\col}(t)\}$
                }
                \If{$p_i.\op$ is $>$ or $\geq$}{
                    $\bX_i \leftarrow \max \{\bX_i, \pi_{p_i.\col}(t)\}$
                }
            }
            \ret $\bX, \bY$
        }
        \myproc{\invertrange{$\bX, \bY$}}{
            $\bY' \leftarrow \bX, \bX' \leftarrow \bY$ \;
            flip $-\infty$ to $\infty$ and $\infty$ to $-\infty$ in $\bY'$ and $\bX'$ respectively. \;
            \ret $\bX', \bY'$
        }
	\caption{{\sc DC verification}}
	\label{algo:main}
\end{algorithm}

\subsection{Verification Algorithm}

In this section we present our verification algorithm that leverages prior work on orthogonal range search. Without loss of generality, we will assume that all predicates of the DC contain only equalities and inequalities but no disequality and that the DC is homogeneous. Both of these assumptions will be removed later. Finally, we assume that the categorical columns in \R{R}\ have been dictionary-encoded to integers, a standard assumption in line with prior work~\cite{pena2019discovery, pena2020efficient}.

\eat{\fp{What is the purpose of lines 4-9? and how it connects with the differences between equalities and range searches?}

\fp{Also add line ranges instead of individual lines}}

Algorithm~\ref{algo:main} describes the details for verifying a homogeneous DC $\varphi$ over a relation \R{R}. On Line~\ref{line:non-eq-count}, we compute the number $k$ of columns that appear in non-equality predicates in $\varphi$. If $\varphi$ contains only equality in all the predicates, then $k=0$. For each tuple $\tup{t}$ in \R{R}, we first project $\tup{t}$ on all columns that participate in an equality predicate (Line~\ref{line:4}) to get $v$. If the projection $v$ has not been seen before, then we insert $v$ in the hash table $H$ and initialize $H[v]$ (Lines~\ref{line:init-begin}-\ref{line:init-end}). Next, we process the projected tuple $v$ based on whether the DC contains only equality predicates or not (Lines~\ref{line:contains-inequality}-\ref{line:end}). If the DC only contains an equality operator in all the predicates, it is sufficient to check if there exist two tuples whose projection over $\vars_{=}(\varphi)$ is equal which would constitute a violation. This is done by storing the count in a hash map which is incremented (Lines~\ref{increment}-\ref{line:end}). If the DC contains a predicate with inequalities, we build a range search data structure of dimension $k$. The $k$ dimensional point inserted into the tree is the tuple obtained by projecting $\tup{t}$ on all non-equality columns (Line~\ref{line:insert}). Before we insert, we check that the new point would not satisfy all the inequality predicates (i.e., form a violation) when grouped with any previously processed point (Lines~\ref{line:check-begin}-\ref{line:insert}). Next, we give an example of how the algorithm works. Figure~\ref{fig:tikz} helps visualizing the ideas behind the example.

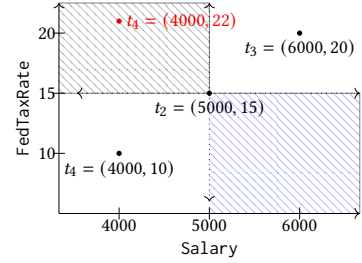
\begin{figure}[!htp]
\scalebox{0.8}{
\begin{tikzpicture}
\draw[->,] (0,0)--(5,0);
\draw[->] (0,0)--(0,3.5);
\node[rotate=90] at (-0.6,1.75) {\texttt{FedTaxRate}};
\node[rotate=0] at (2.5, -0.6) {\texttt{Salary}};

\draw (1,0.1) -- (1,-0.1);
\node at (1,-0.2) {$4000$};
\draw (2.5,0.1) -- (2.5,-0.1);
\node at (2.5,-0.2) {$5000$};
\draw (4,0.1) -- (4,-0.1);
\node at (4,-0.2) {$6000$};

\draw (-0.1,1) -- (0.1,1);
\node at (-0.2,1) {$10$};
\draw (-0.1,2) -- (0.1,2);
\node at (-0.2,2) {$15$};
\draw (-0.1,3) -- (0.1,3);
\node at (-0.2,3) {$20$};

\draw[red, fill=red] (1, 3.2) circle[radius=1pt] node[right] {\color{red} $\tup{t}_4 = (4000, 22)$};
\draw[fill=black] (1, 1) circle[radius=1pt, fill=black] node[below] {$\tup{t}_4 = (4000, 10)$};
\draw[fill=black] (2.5, 2) circle[radius=1pt, fill=black] node[below] {$\tup{t}_2 = (5000, 15)$};
\draw[fill=black] (4, 3) circle[radius=1pt, fill=black] node[below] {$\tup{t}_3 = (6000, 20)$};


\draw[->,dotted] (2.5,2)--(2.5,3.5);
\draw[->,dotted] (2.5,2)--(0.3,2);
\draw[->,dotted] (2.5,2)--(2.5, 0.2);
\draw[->,dotted] (2.5,2)--(5,2);

\draw[fill=gray!80, pattern=north west lines, opacity=0.5] (0,2) -- (2.5,2) -- (2.5,3.5) -- (0,3.5);

\draw[fill=red!80, pattern=north west lines, opacity=0.5, pattern color=blue] (2.5,2) -- (5,2) -- (5, 0) -- (2.5, 0);

\end{tikzpicture}}
\vspace{-4mm}
\caption{\texttt{Salary} and \texttt{FedTaxRate} for each tuple in \texttt{Tax}. The grey (upper left quadrant centered at $t_2$) and blue shaded areas (lower right quadrant centered at $t_2$) show the regions where the tuples that could form a violation $t_2$ lie.}
\vspace{-2mm}
\label{fig:tikz}
\end{figure}

\begin{example} \label{ex:5}
    Consider the \texttt{Tax} table from our running example and the DC $\varphi_3: (\proj{s}{\mathtt{State}} = \proj{t}{\mathtt{State}} \land \proj{s}{\mathtt{Salary}} <  \proj{t}{\mathtt{Salary}} \land \proj{s}{\mathtt{FedTaxRate}} >  \proj{t}{\mathtt{FedTaxRate}})$ which contains one equality and two inequality predicates. Algorithm~\ref{algo:main} will first start with the equality predicate, and place $\tup{t}_1$ in a hash partition by hashing $\tup{t}_1.\mathtt{State} = \text{New York}$. Since the range tree for the hash bucket is empty, the range search will return false and we insert $(\tup{t}_1.\mathtt{Salary},$ $\tup{t}_1.\mathtt{FedTaxRate})$ in the tree. Next, we process $\tup{t}_2$ which is placed in a different partition since $\tup{t}_2.\mathtt{State} = \text{Wisconsin}$. The algorithm performs a range search which returns false since the tree corresponding to that partition is empty. We then insert $(5000, 15)$ in the tree. When $\tup{t}_3$ is processed, it is placed in the same partition as $\tup{t}_2$ since they have the same $\mathtt{State}$ value. {At this point, we have two tuples in the same partition and thus we need to consider the remaining predicates in the DC to establish whether there is a violation. Such a violation would occur in two scenarios: 1) if the tuple already present in the tree ($\tup{t}_2$) has lower \texttt{Salary} than $6000$ but larger than $20$ \texttt{FedTaxRate} which are the corresponding values of \texttt{Salary} and \texttt{FedTaxRate} for tuple $\tup{t}_{3}$, or 2) if $\tup{t}_2$ has salary higher than $6000$ but a \texttt{FedTaxRate} smaller than $20$. To identify whether this is the case, we  perform an orthogonal range search with $\bX = (-\infty, 20)$ and $\bY = (6000, \infty)$ (scenario 1). \eat{In other words, we wish to find if there exists a tuple already present in the tree, such that the first coordinate (i.e. \texttt{Salary}) is smaller than $6000$ but the second coordinate (i.e. \texttt{FedTaxRate}) is larger than $20$. However, it could be the case that there is a tuple in the tree corresponding to a larger salary than $6000$ but a smaller \texttt{FedTaxRate} than $20$.} Then, we also search in the \emph{inverted} range $\bX' = (6000, -\infty)$ and $\bY' = (\infty, 20)$ (scenario 2).} Since $\tup{t}_2$ does not lie in the desired range, both range searches return false and we insert $(6000, 20)$ in the tree. Finally, $\tup{t}_4$ is processed and placed in the same partition as $\tup{t}_2$ and $\tup{t}_3$. We thus query the tree with $\bX = (-\infty, 10)$ and $\bY = (4000, \infty)$ (and the inverted range $\bX' = (4000, -\infty), \bY' = (\infty, 10)$) but no point satisfies the criteria as shown in Figure~\ref{fig:tikz}. Both searches return false, we insert $\tup{t}_4$ in the tree, and return true (Line~\ref{line:return-true}). 

    To demonstrate an example of a violation, consider Table~\texttt{Tax'} with the modified tuple $\tup{t}_4$ with $\mathtt{FedTaxRate} = {\mathbf{22}}$ (shown as $\tup{t}_4$ in red in Figure~\ref{fig:tikz}). Then, the range search queries would be  $\bX = (-\infty, 22), \bY = (4000, \infty)$ and $\bX' = (4000, -\infty), \bY' = (\infty, 22)$. Then, $\tup{t}_2$ and $\tup{t}_3$ form a violation with $\tup{t}_4$ since both the points represent a higher salary than $4000$ but a smaller tax rate than $22$, and Line~\ref{line:return-false-inequality} returns false.\qed
\end{example}

We now establish the correctness of Algorithm~\ref{algo:main}.

\begin{lemma}
    Algorithm~\ref{algo:main} correctly determines whether a homogeneous DC $\varphi$ is satisfied.
\end{lemma}

\eat{\fp{maybe use $t_1, t_2$ instead of $s, t$ here, or vice versa. it becomes confusing when alternating.}}

\begin{proof}    We first show that Algorithm~\ref{algo:main} is correct when $\varphi$ only contains equality predicates. In this case, it is sufficient to determine whether there exist two distinct tuples $\tup{t}_1$ and $\tup{t}_2$ such that $\pi_{\vars_{=}(\varphi)}(\tup{t}_1) = \pi_{\vars_{=}(\varphi)}(\tup{t}_2)$. The hash table $H$ stores a counter for each distinct $\pi_{\vars_{=}(\varphi)}(\tup{t})$ and increments it for each tuple $\tup{t} \in R$ (line \ref{increment}). Thus, the algorithm will correctly return false as soon as some counter becomes greater than one and return true only if no such $t_1, t_2$ exists.

Next, we consider the case when there exists at least one predicate with inequality. We show the proof for the case when all inequality predicate operators are $<$, i.e., all predicates in the DC are of the form $(\tup{s}.A = \tup{t}.A)$ or $(\tup{s}.A < \tup{t}.A)$. The proof for other operators is similar. We first state the following claim.

    \begin{claim} \label{claim}
        Let ${w}$ be the set of attributes that appear in the predicates with inequalities. Two tuples $\tup{t}_1$ and $\tup{t}_2$ in the same partition can form a violation if and only if $\tup{t}_1({w}) \prec \tup{t}_2({w})$ or $\tup{t}_2({w}) \prec \tup{t}_1({w})$, where the notation $\tup{t}({w})$ denotes the projection, $\pi_{{w}}(\tup{t})$, of tuple $t$ on attributes $w$.
    \end{claim}
    
    Here, $\prec$ is the standard coordinate-wise strict dominance checking operator. Claim~\ref{claim} follows directly from the semantics of the operator under consideration and the definition of a violation. Suppose $t$ is the tuple being inserted in the tree. Line~\ref{line:search} will query the range tree with $\bX = (-\infty, \dots, -\infty), \bY = (\tup{t}({v_1}), \dots, \tup{t}({v_k}))$ and $\bX' = (\tup{t}({v_1}), \dots, \tup{t}({v_k})),$ $\bY' = (\infty, \dots, \infty)$. In other words, the algorithm searches for a point in the tree such that $t$ is strictly smaller or larger for each of the $k$ coordinates. The existence of such a point would imply that there exists a pair that forms a violation. 

    If the orthogonal range search finds no point, Claim~\ref{claim} tells us that $\tup{t}$ cannot form a violation with any tuple $\tup{s}$ already present in the range tree. In each iteration of the loop, we insert one tuple into the range tree. Therefore, if $\tup{t}_1$ and $\tup{t}_2$ form a violation, it will be discovered when one of them (say $\tup{t}_2$) is already inserted in the range tree and $\tup{t}_1$ is being processed by the for loop. This completes the proof.
\end{proof}

\introparagraph{Time and Space Complexity} We next establish the running time of the algorithm. First, observe that if $k = 0$, then the algorithm takes $O(|\R{R}|)$ time since the for loop only performs a constant number of hash table operations. If $k \geq 1$, the algorithm performs one insertion and two Boolean orthogonal range search queries in each iteration of the for loop. Suppose the insertion time complexity, denoted by $I(n)$, is of the form\footnote{Throughout the paper, we use $\log^k N$ to mean $(\log N)^k$ and not iterated logarithms.} $\log^\alpha n$ and search time complexity is $T(n)$ when the data structure has $n$ points in it. The running time can be bounded as 

    \begin{align*}
        &\sum \limits_{i = 1}^{|\R{R}|} \big( \underbrace{\log^\alpha i}_{\text{insertion time}} + \underbrace{2 \cdot T(i)}_{\text{query time}} \big) \\
        &< \int_1^{|\R{R}|+1} \log^{\alpha}i\,di + \int_1^{|\R{R}|+1}  2 \cdot T(i) \,di\\
        &= O(|\R{R}| \cdot \log^{\alpha} |\R{R}|) + \int_1^{|\R{R}|+1}  2 \cdot T(i) \,di
    \end{align*}

Seminal work by Overmars~\cite{overmars1983design} showed that using range trees and $k$-d trees, one can design an algorithm with the parameters as shown in~\autoref{table:param}.

\begin{table}
\small
\caption{Data structure parameter on input of size $n$~\cite{overmars1983design}. $k$ is the number of dimensions of the points inserted in the tree.} \label{table:param}
 \begin{tabularx}{\linewidth}{ l l l l} 
 \toprule
 DS & Insertion $I(n)$ & Answering $T(n)$ & Space $S(n)$ \\
 \midrule
Range tree & $O(\log^k n)$ & $O(\log^{k} n)$  &  $O(n \cdot \log^{k-1} n)$ \\
$k$-d tree & $O(\log n)$ & $O(n^{1 - \frac{1}{k}})$  &  $O(n)$ \\

\bottomrule
 \end{tabularx}
 \end{table}

The integral in the second term in the equation above can be bounded by setting $T(i) = \log^{k} i$ or $T(i) = i^{1-1/k}$. In both cases, the second term evaluates to $O(|\R{R}| \cdot T(|\R{R}|))$. For space usage, note that the hash table takes a linear amount of space in the worst case. Thus, the space requirement of the tree data structure determines the space complexity. The main result can be stated as follows.

\begin{theorem} \label{thm:main}
    Algorithm~\ref{algo:main} runs in time $O(|\R{R}| \cdot (I(|\R{R}|) + T(|\R{R}|)))$ and uses space $S(|\R{R}|)$ when using range tree or $k$-d tree with parameters as shown in~\autoref{table:param}.
\end{theorem}

With range trees, the running time is $O(|\R{R}| \cdot \log^k |\R{R}|)$ and space usage is $O(|\R{R}| \log^{k-1} |\R{R}|)$; for $k$-d trees, the running time is $O(|\R{R}|^{2 - \frac{1}{k}})$ and space requirement is $O(|\R{R}|)$.

\introparagraph{Comparison with \textsc{Facet}} Our approach is superior to \textsc{Facet} in three respects. First, we use polynomially less space and time in the worst-case. \eat{Second, even on instances that are the worst-case, the performance can vary significantly as shown below.}
Second, there exist instances where our algorithm saves a significant amount of time and space by early termination.

\begin{proposition}
    For every homogeneous DC $\varphi$ with at least one non-equality predicate, there exists a relation $\R{R}$ such that Algorithm~\ref{algo:main} takes $O(1)$ time and \textsc{Facet} requires $\Omega(|\R{R}|)$ time.
\end{proposition}
\begin{proof}
    We sketch the proof for $\varphi : \forall \tup{s}, \tup{t} \in \R{R},\ \neg (\tup{s}.A = \tup{t}.A \land \tup{s}.B < \tup{t}.B)$ which can be extended straightforwardly for other DCs of interest. We construct a unary relation $\R{R}(A, B)$ of size $N$ as follows: the first tuple $t_1$ is $(a_1, b_1)$ and the remaining $N-1$ tuples are $(a_1, b_2)$ where $b_1 < b_2$. Note that $t_1$ forms a violation with every other tuple in the relation. Algorithm~\ref{algo:main} initializes one range tree when processing $t_1$ (Line~\ref{line:5}) and inserts $t_1$ in it (Line~\ref{line:insert}). 
    Thereafter when tuple $t_2$ is processed, the range search query (Line~\ref{line:search}) will return true and the algorithm will terminate. Note that all the operations take $O(1)$ time since the tree only contains two tuples. However, \textsc{Facet} requires $\Omega(|\R{R}|)$ time for processing the refinement of $\tup{s}.A = \tup{t}.A$ predicate.
\end{proof}

\smallskip
\noindent Lastly, the space requirement of \textsc{Facet} is relation dependent. If the machine has only linear amount of memory, \textsc{Facet} will be unable to complete the refinements and fail. On the other hand, our solution allows verification with linear space using $k$-d trees. This flexibility is important for resource-constrained production scenarios.

\subsection{Generalizations and Optimizations} \label{sec:opt}

In the previous section, we made some assumptions on the type of constraints processed by Algorithm~\ref{algo:main}. We now gradually remove these assumptions and present appropriate examples and proofs.

\smallskip
\introparagraph{Allowing inequality heterogeneous predicates} We first extend our algorithm to also handle 
heterogeneous predicates, namely predicates of the form $\proj{s}{\A{A}}\ \op \ \proj{t}{\A{B}}$, where
$\op$ is $<, \leq, >$ or $\geq$.
Let $\varphi$ be a DC containing row-level homogeneous predicates (as before) and some inequality
heterogeneous predicates.
The main difference from the previous case is that we now need to generalize our procedure for
computing the ranges $(\bX,\bY)$ and inverted ranges $(\bX',\bY')$ for range search. 
Algorithm~\ref{algo:rangesearch} shows the generalization. The main idea is that
if $\varphi$ has a predicate $s.C < t.D$, then when we process a new tuple $r$, the upper-bound
for {\em{attribute $C$}} is set to $r.D$ in the forward check, and 
the lower-bound for attribute $D$ is set to $r.C$ in the inverted check (because we are comparing
attribute $C$ of $s$ with attribute $D$ of $t$ in the predicate). When $C = D$, we recover our original
algorithm. We also note that the new generalization also extends our algorithm to handle the case when 
attributes occur in more than one predicate.  Thus, a heterogeneous equality, $s.C = t.D$, can be handled
by rewriting it to $s.C \leq t.D \wedge s.C \geq t.D$ and using the generalized range computation from
Algorithm~\ref{algo:rangesearch}. The range trees store projections of tuples on the attributes that are
involved in inequality predicates.

Heterogeneity also enables an optimization. Let $L_1$ be all the attributes present in inequality predicates
and referenced by $s$, and $L_2$ be those that are referenced by $t$. For example, if $s.C < t.D$ is a predicate
in $\varphi$, then $L_1$ will contain $C$ and $L_2$ will contain $D$. Now, rather than having one range search
data structure of dimension $|L_1\cup L_2|$, we can instead have two potentially smaller range search
data structures of dimension $|L_1|$ and $|L_2|$ -- one to perform the forward search and the other
to perform the inverted search. In absence of heterogeneous predicates, we had $L_1=L_2$ and both these were
identical, but in presence of heterogeneous constraints, $L_1$ and $L_2$ can each be strictly smaller 
than their union.

\begin{example}
    Consider the DC $\varphi: \neg (\proj{s}{\mathtt{Salary}}\ \leq \ \proj{t}{\mathtt{FedTaxRate}})$. Note that $L_1 = \{\mathtt{Salary}\}$ and $L_2 = \{\mathtt{FedTaxRate}\}$. Suppose we are processing tuple $r \in 
    \R{R}$. We will create two range search data structures $H_1$ (in which we will store $t.\mathtt{Salary}$) and $H_2$ (in which we will store $t.\mathtt{FedTaxRate}$). Given $r$, we first do a range search in $H_1$ to check if there is a point that is no larger than $r.\mathtt{FedTaxRate}$. If there is a point, we have found a violation. Otherwise, we insert $r.\mathtt{Salary}$ in $H_1$ and check whether there is a point in $H_2$ that is no smaller than $r.\mathtt{Salary}$. If there is a point, we have found a violation and we insert $r.\mathtt{FedTaxRate}$ into $H_2$ otherwise.\qed
    \eat{When constructing the tree structure, we insert the one-dimensional point $\proj{r}{\mathtt{Salary}}$. Suppose $t_1.\mathtt{Salary} = 3000$ has already been inserted and $t_2$ is being processed. The range query will be $\bX = (-\infty), \bY = (15)$ (15 is the value of \texttt{FedTaxRate} for $t_{2})$. In words, this range query corresponds to finding a point in the tree (i.e. a \texttt{Salary}) that has a value smaller than $15$. Since we are inserting \texttt{Salary} into the tree, if the range query finds a point smaller than $15$, it would correspond to a tuple that has \texttt{Salary} smaller than $15$, thus discovering a violation.}
\end{example}

\begin{algorithm}[!t]
    \small
	\SetCommentSty{textsf}
	\DontPrintSemicolon 
	\SetKwInOut{Input}{Input}
	\SetKwInOut{Output}{Output}
	\SetKwRepeat{Do}{do}{while}
	\SetKwFunction{rangeSearch}{\textsc{booleanRangeSearch}}
	\SetKwFunction{ins}{\textsc{insert}}
        \SetKwFunction{createsearchrange}{\textsc{CreateBothSearchRanges}}
        \SetKwFunction{invertrange}{\textsc{InvertRange}}
	\SetKwProg{myalg}{procedure}{}{}
	\SetKwData{tree}{\textsc{OrthogonalRangeSearch}}
	\SetKwData{op}{$\mathsf{op}$}
        \SetKwData{nonEqCount}{$\mathsf{k}$}
	\SetKwData{col}{$\mathsf{col}$}
	\Input{A tuple $r$ from relation $\R{R}$, DC $\varphi$}
        \Output{Search range and inverted search range}
	\SetKwProg{myproc}{\textsc{procedure}}{}{}
	\SetKwData{ret}{\textbf{return}}
        \myproc{\createsearchrange{$r$, $\varphi$}}{
            $\bX \leftarrow (-\infty, \dots, -\infty), \bY \leftarrow (\infty, \dots, \infty)$ \tcc*{$\bX,\bY,\bX',\bY'$ are indexed by attributes of $R$ that appear in inequality predicates}
            $\bX' \leftarrow (-\infty, \dots, -\infty), \bY' \leftarrow (\infty, \dots, \infty)$ \;
            \ForEach{\text{inequality predicate } $s.C \ \op \ t.D$ \text{ in $\varphi$}}{
                \If{$\op$ is $<$ or $\leq$}{
                    $\bY.C \leftarrow \min \{\bY.C, r.D\}$\;
                    $\bX'.D \leftarrow \max \{\bX'.D, r.C\}$
                }
                \If{$\op$ is $>$ or $\geq$}{
                    $\bX.C \leftarrow \max \{\bX.C, r.D\}$ \;
                    $\bY'.D \leftarrow \min \{\bY'.D, r.C\}$ 
                }
            }
            \ret $\bX, \bY, \bX', \bY'$
        }
	\caption{{\sc Search range generation when processing a new tuple $r$ in presence of heterogeneous inequalities}}
	\label{algo:rangesearch}
\end{algorithm}

\eat{
DCs containing row-level homogeneous predicates and heterogeneous predicates (which are of the form $\proj{s}{\A{A}}\ \op \ \proj{t}{\A{B}}$) can also be processed using range search queries but requires a slightly different range query building technique and two range search data structures. We insert the tuple formed by all columns present in a heterogeneous predicate referenced by $s$ (let's refer these columns as $L_1$) into one data structure and all columns present in a heterogeneous predicate referenced by $t$ (say $L_2)$ into another. When creating the range query tuples $\bX, \bY$ for some $r \in \R{R}$, we use $r(L_1)$ and $r(L_2)$ for doing the range checks in both the data structures. Special care needs to be taken about how we build the trees incrementally. We illustrate with the help of an example. 
\begin{example}
    Consider the DC $\varphi: \neg (\proj{s}{\mathtt{Salary}}\ \leq \ \proj{t}{\mathtt{FedTaxRate}})$. Note that $L_1 = \{\mathtt{Salary}\}$ and $L_2 = \{\mathtt{FedTaxRate}\}$ Suppose we are processing tuple $r \in 
    \R{R}$. We will create two range search data structures $H_1$ (in which we will store $t.\mathtt{Salary}$) and $H_2$ (in which we will store $t.\mathtt{FedTaxRate}$). Given $r$, we first do a range search in $H_1$ to check if there is a point that is no larger than $r.\mathtt{FedTaxRate}$. If there is a point, we have found a violation. Otherwise, we insert $r.\mathtt{Salary}$ in $H_1$ and check whether there is a point in $H_2$ that no smaller than $r.\mathtt{Salary}$. If there is a point, we have found a violation and we insert $r.\mathtt{FedTaxRate}$ into $H_2$ otherwise.
    
    \eat{When constructing the tree structure, we insert the one-dimensional point $\proj{r}{\mathtt{Salary}}$. Suppose $t_1.\mathtt{Salary} = 3000$ has already been inserted and $t_2$ is being processed. The range query will be $\bX = (-\infty), \bY = (15)$ (15 is the value of \texttt{FedTaxRate} for $t_{2})$. In words, this range query corresponds to finding a point in the tree (i.e. a \texttt{Salary}) that has a value smaller than $15$. Since we are inserting \texttt{Salary} into the tree, if the range query finds a point smaller than $15$, it would correspond to a tuple that has \texttt{Salary} smaller than $15$, thus discovering a violation.}
\end{example}





\smallskip
\noindent Next, we remove the assumption about homogeneous DC not having any disequality operator predicate, followed by an optimization that improves the runtime for verifying DCs with special structure.
}

\smallskip
\introparagraph{Allowing disequality predicates} Any predicate $\proj{s}{\A{A}}\ \neq\ \proj{t}{\A{B}}$ can be written as a union of two predicates: $(\proj{s}{\A{A}}\ <\ \proj{t}{\A{B}}) \lor (\proj{s}{\A{A}}\ >\ \proj{t}{\A{B}})$. Therefore, a DC containing $\ell$ predicates with $\op$ as $\neq$ can be equivalently written as a conjunction of $2^{\ell}$ DCs containing no disequality operator.

If the original homogeneous DC contains no inequality predicate, then it is possible to reduce the number of equivalent DCs from $2^{\ell}$ to $2^{\ell-1}$. The idea is that a violation $(\tup{s}, \tup{t})$ is symmetric (i.e. $(\tup{t}, \tup{s})$ is also a violation) if the DC contains only equality and disequality predicates. Therefore, when converting a DC to only have inequalities, it suffices to expand $(\tup{s}.A \ \neq \ \tup{t}.A)$ to just $(\tup{s}.A \ < \ \tup{t}.A)$ for one last disequality predicate instead of $(\tup{s}.A \ < \ \tup{t}.A) \lor (\tup{s}.A \ > \ \tup{t}.A)$.

\begin{proposition}
    Given a homogeneous DC $\varphi$ with only equality and $\ell$ disequality predicates, there exists an equivalent conjunction of $2^{\ell-1}$ DCs that contain only equality and inequality predicates.
\end{proposition}
\begin{proof}
Consider the constraint $\varphi: \neg (\phi \land \tup{s}.A \ \neq \ \tup{t}.A)$, where $\phi$ is a conjunction of homogeneous equality and disequality predicates. Let $(q, r)$ be a violation to $\varphi$. Without loss of generality, we assume that $\tup{q}.A < \tup{r}.A$, and then $(q, r)$ is also a violation to $\varphi_1: \neg (\phi \land \tup{s}.A \ < \ \tup{t}.A)$. Since $\phi$ only contains equality and disequality predicates, $(r, q)$ also satisfies $\phi$ by symmetricity, and therefore $(r, q)$ is a violation to $\varphi$ and $\varphi_2: \neg (\phi \land \tup{s}.A \ > \ \tup{t}.A)$. In fact, for any violation $(r, q)$ to $\varphi$, one of $(r, q)$ and $(q, r)$ must violate $\varphi_1$ while the other violates $\varphi_2$. Thus, we only need to check $\varphi_1$ for violations, which contains $l-1$ disequality predicates and can be written as a conjunction of $2^{l-1}$ DCs containing no disequality predicates by logical equivalence.
    \eat{Consider a violation $(\tup{s}, \tup{t})$ for $\varphi$. Without loss of generality, suppose $\tup{s} \prec \tup{t}$ on the $\ell$ coordinates. This violation will be detected by the DC where the $\op$ is $<$ for all inequality predicates. However, since $(\tup{t},\tup{s})$ is also a violation, the DC where the $\op$ is $>$ for all inequality predicates will detect $(\tup{t},\tup{s})$, and thus redundant. In general, since $\ell-1$ predicates have been expanded, every violation when considering the $\ell-1$ predicates will be detected by two DCs but the symmetric property ensures that for the $\ell^{th}$ predicate, only one operator (say $<$) is sufficient.}
\end{proof}

\smallskip
\noindent Note that although Algorithm~\ref{algo:main} is described for a single homogeneous DC, it can be extended to verify multiple DCs in a non-serial fashion. For each tuple in $\R{R}$, one can perform the processing for each DC and return as soon as any of them detect a violation. This ensures we retain the ability to terminate early whenever possible.

\smallskip
\introparagraph{Allowing only one inequality predicate} If a DC has row homogeneous equality predicates and at most one predicate (homogeneous or heterogeneous) containing an inequality, then the verification can be done in linear time. Algorithm~\ref{algo:one-inequality} shows the algorithm. Like the previous algorithms, we begin by partitioning the input into a hash table based on the equality predicates. Let the inequality predicate be $\proj{s}{\A{A}} \ \op \ \proj{t}{\A{B}}$. The main idea is to maintain the running minimum and maximum values for the Columns $\A{A}$ and $\A{B}$ {\em{for each partition of the input}}.
Since the comparison is one-dimensional, it is sufficient to compare against the minimum (or maximum) value.
Lines~\ref{line:one-inequality-init-begin}-\ref{line:one-inequality-init-end} initialize these minimum and maximum values to $+\infty$ and $-\infty$ respectively when we see a tuple that belongs to no existing partition.
Lines~\ref{line:one-inequality-check-begin}-\ref{line:one-inequality-check-end} then perform the inequality check between all previously seen tuples (in this new tuple's partition $v$) and this new tuple. 
If the checks fail, we update the minimum and maximum values for partition $v$ based on this new tuple on Lines~\ref{line:minA}-\ref{line:maxB}.
The algorithm makes only one pass over the entire dataset and the overall time complexity is $O(|\R{R}|)$. While this optimization is simple, it has important implications. In particular, popular constraints such as functional dependencies (FD) are DCs that contain exactly one inequality predicate. Algorithm~\ref{algo:one-inequality} recovers the standard linear time algorithm to verify FDs~\cite{ibaraki1999functional}. However, it is unclear that \facet, in its present form, can achieve the same provable guarantee.

\begin{algorithm}[!t]
    \small
	\SetCommentSty{textsf}
	\DontPrintSemicolon 
	\SetKwInOut{Input}{Input}
	\SetKwInOut{Output}{Output}
	\SetKwRepeat{Do}{do}{while}
	\SetKwFunction{rangeSearch}{\textsc{booleanRangeSearch}}
	\SetKwFunction{ins}{\textsc{insert}}
        \SetKwFunction{lowerrange}{\textsc{SearchRange}}
	\SetKwProg{myalg}{procedure}{}{}
	\SetKwData{tree}{\textsc{OrthogonalRangeSearchTree}}
	\SetKwData{op}{$\mathsf{op}$}
        \SetKwData{nonEqCount}{$\mathsf{k}$}
        \SetKwData{min}{$\mathsf{min}$}
        \SetKwData{max}{$\mathsf{max}$}
        \SetKwFunction{List}{\textsc{list}}
        \SetKwFunction{last}{\textsc{last}}
        \SetKwFunction{first}{\textsc{first}}
	\SetKwData{col}{$\mathsf{col}$}
	\Input{Relation $\R{R}$, DC $\varphi$ containing equality predicates of form $\proj{s}{\A{C}}=\proj{t}{\A{C}}$ (for some $C$) and {\bf{one}} inequality predicate $p$ of form $\proj{s}{\A{A}}\ \op \ \proj{t}{\A{B}}$}
        \Output{True/False}
	\SetKwProg{myproc}{\textsc{procedure}}{}{}
	\SetKwData{ret}{\textbf{return}}
        \SetKwData{flag}{isViolation}
        $\min_{\A{A}}, \max_{\A{A}}, \min_{\A{B}}, \max_{\A{B}}\leftarrow $ empty hash table\;
	\ForEach{$t \in R$}{
            $v \leftarrow \pi_{\vars_{=}(\varphi)} (t)$ \;
            \If{$v \notin \min_{\A{A}}$}{ \label{line:one-inequality-init-begin}
                $\min_{\A{A}}[v], \min_{\A{B}}[v] \leftarrow +\infty, +\infty$\;
                $\max_{\A{A}}[v], \min_{\A{B}}[v] \leftarrow -\infty, -\infty$\label{line:one-inequality-init-end}
            }
            \If{$(p.\op \in \{<, \leq\} \land \min_{\A{A}}[v] \ \op \ t.\A{B}) \lor (p.\op \in \{>, \geq\} \land \max_{\A{A}}[v] \ \op \ t.\A{B})$ \label{line:one-inequality-check-begin}}{
                    \ret \textbf{false}
                }
            \If{$(p.\op \in \{<, \leq\} \land t.\A{A} \ \op \ \max_{\A{B}}[v]) \lor (p.\op \in \{>, \geq\} \land t.\A{A} \ \op \ \min_{\A{B}}[v])$ \label{line:reverse-if}}{
                    \ret \textbf{false} \label{line:one-inequality-check-end}
                }
            $\min_{\A{A}}[v] \leftarrow \min \{\min_{\A{A}}[v], t[\A{A}]\}$ \; \label{line:minA}
            $\min_{\A{B}}[v] \leftarrow \min \{\min_{\A{B}}[v], t[\A{B}]\}$ \; \label{line:minB}
            $\max_{\A{A}}[v] \leftarrow \max \{\max_{\A{A}}[v], t[\A{A}]\}$ \; \label{line:maxA}
            $\max_{\A{B}}[v] \leftarrow \max \{\max_{\A{B}}[v], t[\A{B}]\}$ \; \label{line:maxB}
        }
                
        \ret \textbf{true} \;
	\caption{{\sc DC verification for DCs with row-homogeneous equality and one inequality predicate}}
	\label{algo:one-inequality}
\end{algorithm}

\eat{\smallskip
\introparagraph{Inequality predicate} If a DC has at most one predicate (homogeneous or heterogeneous) containing an inequality, then the verification can be done in linear time. Algorithm~\ref{algo:het} shows the main steps when the inequality predicate is heterogeneous. The algorithm for homogeneous inequality predicate is very similar. Like the previous algorithms, we begin by partitioning the input into a hash table based on the equality predicates. Lines~\ref{line:min}-\ref{line:max} store the min and max values for the first column in the inequality predicate. The key insight is that any tuple pair $(s, t)$ in the same hash bucket forms a violation when $\proj{s}{\A{A}}\ \op \ \proj{t}{\A{B}}$ is true. Let us fix $\op$ to $<$. Then, if $\proj{t}{\A{B}}$ is greater than the min value in column $\A{A}$, we can find  such that. This is precisely the condition checked on line~\ref{line:if}. The algorithm makes only three passes for the entire dataset and the overall time complexity is $O(|\R{R}|)$. While this optimization is simple, it has important implications. In particular, popular constraints such as functional dependencies (FD) are DCs that contain exactly one inequality predicate. Algorithm~\ref{algo:het} recovers the standard linear time algorithm to verify FDs~\cite{ibaraki1999functional} and candidate keys. However, it is not clear that prior work~\cite{pena2021fast} can achieve the same provable guarantee.

\begin{algorithm}[!t]
	\SetCommentSty{textsf}
	\DontPrintSemicolon 
	\SetKwInOut{Input}{Input}
	\SetKwInOut{Output}{Output}
	\SetKwRepeat{Do}{do}{while}
	\SetKwFunction{rangeSearch}{\textsc{booleanRangeSearch}}
	\SetKwFunction{ins}{\textsc{insert}}
        \SetKwFunction{lowerrange}{\textsc{SearchRange}}
	\SetKwProg{myalg}{procedure}{}{}
	\SetKwData{tree}{\textsc{OrthogonalRangeSearchTree}}
	\SetKwData{op}{$\mathsf{op}$}
        \SetKwData{nonEqCount}{$\mathsf{k}$}
        \SetKwData{min}{$\mathsf{min}$}
        \SetKwData{max}{$\mathsf{max}$}
        \SetKwFunction{List}{\textsc{list}}
        \SetKwFunction{last}{\textsc{last}}
        \SetKwFunction{first}{\textsc{first}}
	\SetKwData{col}{$\mathsf{col}$}
	\Input{Relation $\R{R}$, DC $\varphi$ with some equality and one heterogeneous inequality predicate $p: \proj{s}{\A{A}}\ \op \ \proj{t}{\A{B}}$}
        \Output{True/False}
	\SetKwProg{myproc}{\textsc{procedure}}{}{}
	\SetKwData{ret}{\textbf{return}}
        \SetKwData{flag}{isViolation}
        $H \leftarrow $ empty hash table \;
	\ForEach{$t \in R$}{
            $v \leftarrow \pi_{\vars_{=}(\varphi)} (t)$ \;
            \If{$v \not\in H$}{
                    $H[v] \leftarrow \text{new } \List()$
            }
            $H[v].\ins(t)$ \\
            \tcc{Let $b$ be the bucket to which $t$ is hashed}
                    $\min^b_{\A{A}} \leftarrow \min \{\min^b_{\A{A}}, t[\A{A}]\}$ \; \label{line:min}
                    $\max^b_{\A{A}}\leftarrow \max \{\max^b_{\A{A}}, t[\A{A}]\}$ \; \label{line:max}
        }
        \ForEach{$t \in R$}
        {
            \If{$(p.\op \in \{<, \leq\} \land \min^b_{\A{A}} \ \op \ t.\A{B}) \lor (p.\op \in \{>, \geq\} \land \max^b_{\A{A}} \ \op \ t.\A{B})$ \label{line:if}}{
                    \ret \textbf{false}
                }
        }
                
        \ret \textbf{true} \;
	\caption{{\sc DC verification for Heterogeneous predicate}}
	\label{algo:het}
\end{algorithm}
}

\smallskip
\introparagraph{Allowing mixed homogeneous constraints}
We now extend our verification algorithm to work also for mixed homogeneous constraints that
can contain predicates of the form $s.A\,\op\,s.B$ as well as $s.A\,\op\, t.A$.
Let $\forall{s,t}:\neg \phi(s,t)$ be a mixed homogeneous denial constraint.
We first rewrite $\phi$ in the form $\phi_S(s) \wedge \phi_T(t) \wedge \phi_{ST}(s,t)$
where 
$\phi_S$ contains all predicates that mention only $s$ (and not $t$),
$\phi_T$ contains all predicates that mention only $t$ (and not $s$),
and 
$\phi_{ST}$ contains all predicates that mention both $s$ and $t$.
The constraint $\forall{s,t}:\neg\phi$ that we need to verify over a given $\R{R}$ can be 
equivalently rewritten as follows:
\begin{eqnarray*}
\forall{s,t}: \neg \phi & \Leftrightarrow &
\forall{s,t}: \neg (\phi_S(s) \wedge \phi_T(t)) \lor \neg \phi_{ST}(s, t)
\\ & \Leftrightarrow & 
\forall{s,t}: (\phi_S(s) \wedge \phi_T(t)) \Rightarrow \neg \phi_{ST}(s, t)
\\ & \Leftrightarrow & 
\forall{s\in \R{S}}:\forall{t\in \R{T}}: \neg \phi_{ST}(s, t)
\end{eqnarray*}
where $\R{S}$ is the set of all tuples in $\R{R}$ s.t. $\phi_S$ is true,
and $\R{T}$ is the set of all tuples in $\R{R}$ s.t. $\phi_T$ is true. Note that $\R{S}$ and $\R{T}$ can overlap.

We maintain two range search data structures (same as the $H$ in Algorithm~\ref{algo:main}) $H_{\R{S}}$ and $H_{\R{T}}$ for points in $\R{S}$ and $\R{T}$ respectively.
For each tuple (aka point) $q \in \R{R}$, we first check whether it belongs to $\R{S}$ and $\R{T}$.
\begin{enumerate}
    \item If $q \in \R{S}$, we perform range search on $H_{\R{T}}$ to find any point $r$ such that $\phi_{ST}(q, r)$ is true. If there is no such point, we insert $q$ into $H_{\R{S}}$. Otherwise, the constraint does not hold and the algorithm terminates. This step checks whether there is any previously seen point $r$ in $\R{T}$ such that $(q, r)$ forms a violation.
    \item Similarly, if $q \in \R{T}$, we perform range search on $H_{\R{S}}$ to find any point $r$ such that $\phi_{ST}(r, q)$ is true. If there is no such point, we insert $q$ into $H_{\R{T}}$. Otherwise, we output false.
\end{enumerate}

The correctness of the algorithm follows from the logical equivalences established above.
Whenever there exist $s \in \R{S}, t \in \R{T}$ such that $\phi_{ST}(s, t)$ is true, the algorithm is able to identify them no matter whether $s$ precedes $t$ or not in the input relation.
For each point, the algorithm performs at most two range queries, and therefore the big-O complexity is the same as the original algorithm.

\eat{Let $\R{S}^c$ denote the complement of $\R{S}$, that is, all tuples in $\R{R}$ that are not in $\R{S}$. 
Next, let us define three sets. 
The set $ST'$ is defined as $\R{S}\cap \R{T}^c$,
$S'T$ is $\R{S}^c \cap \R{T}$ and $ST$ is $\R{S} \cap \R{T}$.
Now we further rewrite our constraint equivalently as follows:
\begin{eqnarray*}
\forall{s\in \R{S}}: \forall{t\in \R{T}}: \neg \phi_{st}
& \Leftrightarrow &
\forall{s\in ST'\cup ST}: \forall{t\in S'T\cup ST}: \neg \phi_{st}
\\
& \Leftrightarrow &
\forall{s\in ST'}: \forall{t\in S'T}: \neg \phi_{st}
\\ &\wedge & 
\forall{s\in ST'}: \forall{t\in ST}:  \neg \phi_{st}
\\ & \wedge &
\forall{s\in ST}: \forall{t\in S'T}:  \neg \phi_{st}
\\ & \wedge &
\forall{s\in ST}: \forall{t\in ST}:  \neg \phi_{st}
\end{eqnarray*}
We can thus check the original DC by verifying the four subproblems. These four checks 
can all be performed in a single iteration by maintaining {\em{three}} range search data structures $H_1, H_2, H_3$, which will contain tuples from $ST'$, $S'T$ and $ST$ respectively.}
\eat{$H_1$ will contain tuples from $ST'$, $H_2$ will contain from $S'T$ and $H_3$ will contain from $ST$. 
When a new tuple $s$ is processed, we first determine if it belongs to $ST'$, $S'T$, or $ST$.
\\
(1) If $s\in ST'$, we check
$\forall{t\in S'T}:  \phi_{st}(s,t)$ using range search on points in $H_2$ and we check
$\forall{t\in ST}:  \phi_{st}(s,t)$ using range search on points in $H_3$, and if both succeed, we insert $s$ in $H_1$.
\\
(2) If $s\in S'T$, we check
$\forall{t\in ST'}:  \phi_{st}(t,s)$ (note the inverted check) using range search on points in $H_1$ and we check
$\forall{t\in ST}:  \phi_{st}(t,s)$ (note again the inverted check) using range search on points in $H_3$, and 
if both succeed, we insert $s$ in $H_2$.
\\
(3) If $s\in ST$, we check
$\forall{t\in S'T}:  \phi_{st}(s,t)$ using our algorithm on $H_1$ and we check
$\forall{t\in ST}:  \phi_{st}(s,t)\wedge\phi_{st}(t,s)$ (note both the forward and inverted check) using range search on points in $H_3$, and
if both succeed, we insert $s$ in $H_3$.
\\}
\eat{The correctness of the algorithm for mixed homogeneous DC verification follows from the logical equivalences established above. Since $ST'$, $S'T$ and $ST$ are disjoint sets and their union contains no more than $|\R{R}|$, our complexity is no more than a constant factor $3$ of the original.}

\eat{
\begin{example}
    Consider the constraint $\phi: \neg (s.\mathtt{SSN} = t.\mathtt{SSN}\ \land\  s.\mathtt{Salary} \geq s.\mathtt{FedTaxRate})$. Here, $\phi_s = (s.\mathtt{Salary} \leq s.\mathtt{FedTaxRate}), \phi_t = \emptyset, \phi_{st} = 
    (s.\mathtt{SSN} = t.\mathtt{SSN}$). This gives us $\R{S} = \R{R},\R{T} = \R{R}$ and thus, $ST' = \emptyset, S'T = \emptyset, ST = \R{R}$. Now, we only need to verify the fourth subproblem where we need to check $\forall s \in \R{R}: \forall t \in \R{R} : \neg (s.\mathtt{SSN} = t.\mathtt{SSN})$, which the reader can verify that it evaluates to true.
\end{example}
}

\eat{\subsection{Two Variants for Verification} 

\begin{algorithm}[!t]
	\SetCommentSty{textsf}
	\DontPrintSemicolon 
	\SetKwInOut{Input}{Input}
	\SetKwInOut{Output}{Output}
	\SetKwRepeat{Do}{do}{while}
	\SetKwFunction{rangeSearch}{\textsc{booleanRangeSearch}}
	\SetKwFunction{ins}{\textsc{insert}}
        \SetKwFunction{lowerrange}{\textsc{SearchRange}}
        \SetKwFunction{invertrange}{\textsc{InvertRange}}
        \SetKwFunction{len}{\textsc{len}}
	\SetKwProg{myalg}{procedure}{}{}
	\SetKwData{tree}{\textsc{OrthogonalRangeSearchTree}}
	\SetKwData{op}{$\mathsf{op}$}
        \SetKwData{nonEqCount}{$\mathsf{k}$}
	\SetKwData{col}{$\mathsf{col}$}
	\Input{Relation $\R{R}$, Homogeneous DC $\varphi$}
        \Output{True/False}
	\SetKwProg{myproc}{\textsc{procedure}}{}{}
	\SetKwData{ret}{\textbf{return}}
        $\nonEqCount \leftarrow |\vars(\varphi) \setminus \vars_{=}(\varphi)|$ \;
        $H,T \leftarrow $ empty hash table \;
	\ForEach{$t \in R$ \label{line:for}}{
            $v \leftarrow \pi_{\vars_{=}(\varphi)} (t)$ \;
            \If{$v \not\in H$}{
                $H[v] \leftarrow $ new List()
            }
            $H[v].\ins(\pi_{\vars(\varphi) \setminus \vars_{=}(\varphi)}(t))$ \;
            \If{$\nonEqCount = 0 \land H[v].\len() > 1$}{
                \ret \textbf{false}        
            }
        }
        \ForEach{bucket $b \in H$}{
            $T[b] \leftarrow $ create \tree\ with all points in $H[b]$ as input \label{line:build}\;
            \ForEach{$t \in H[b]$}{
                $\bX, \bY \leftarrow \lowerrange{t}$ \tcc*{From Algorithm~\ref{algo:main}}
                \If{$T[b].\rangeSearch(\bX, \bY)$ \label{line:query}}{
                        \ret \textbf{false}
                }
            }
            free space used by $T[b]$
        }
        \ret \textbf{true}
	\caption{{\sc Batched DC verification}}
	\label{algo:main2}
\end{algorithm}

\introparagraph{Batched Verification (variant one)} Algorithm~\ref{algo:main} interleaves insertion and querying of the data structure. However, it is also possible to separate the two, i.e., first construct the entire data structure (by \emph{batching} all inserts together), and then begin querying. Algorithm~\ref{algo:main2} shows the batched version. It begins similarly to Algorithm~\ref{algo:main} where we partition the input into hash buckets by using the projection on equality predicates as the key. The key difference is that instead of building the tree and querying \emph{on-the-fly} in each iteration, we first build the tree for all tuples (line~\ref{line:build}), and then query it in a separate loop (line~\ref{line:query}). The correctness readily follows from \autoref{thm:main}. 

\smallskip
To bound the running time, we use the following two seminal results by Bentley et al.~\cite{bentley1979data}

\begin{theorem}[\cite{bentley1979data}] \label{thm:range}
    Given $N$ k-dimensional points, there exists a data structure, known as a \textbf{range tree}, that can be built in pre-processing time $P(N) = O(N \log^{k-1} N)$, using $S(N) = O(N \log^{k-1} N)$ space, and can answer any Boolean orthogonal range search query in time $T(N) = O(\log^{k-1} N)$.
\end{theorem}

\begin{theorem}[\cite{bentley1979data}] \label{thm:kd}
    Given $N$ k-dimensional points, there exists a data structure, known as a \textbf{k-d tree} that can be built in pre-processing time $P(N) = O(N \log N)$, using $S(N) = O(N)$ space, and can answer any Boolean orthogonal range search query in time $T(N) = O(N^{1-\frac{1}{k}})$.
\end{theorem}

The running time can be bounded similarly to \autoref{thm:main} by observing that the for loop on \autoref{line:for} takes time $O(|\R{R}|)$, range tree construction takes time $P(|\R{R}|)$ and we perform one range query for each tuple in the input. Thus, the total time taken is $P(|\R{R}|) + \sum \limits_{i=1}^{|\R{R}|} T(|\R{R}|)$. Since we free up space after processing each partition, the space requirement is dominated by the partition that contains the largest number of tuples.

\begin{theorem} \label{thm:main2}
    Let $P(n)$ denote the construction time, $S(n)$ the space requirement, and $T(n)$ denote the query time of an orthogonal range search data structure consisting of $n$ points. Then, Algorithm~\ref{algo:main2} runs in time $O(P(|\R{R}|) + |\R{R}| \cdot T(|\R{R}|))$ and uses space $S(|\R{R}|)$.
\end{theorem}

\eat{\smallskip
\introparagraph{Improving Batched DC verification} Algorithm~\ref{algo:main2} can be further optimized by shaving off another log factor from the running time. The key idea is that we can sort the relation $\R{R}$ on one of the columns (say $C^\star$) that participates in a predicate with inequality. The comparator function used for sorting is $\op$ in the predicate. For $<$ and $>$ operators, we also keep track of whether there is at least one tuple pair that satisfies the predicate by keeping track of the minimum and maximum value of column $C^\star$.

When creating the range tree, we can now ignore column $C^\star$, thus reducing the dimensions of the points inserted in the range tree by one. \TODO{add statement and proof?}

\begin{algorithm}[!t]
	\SetCommentSty{textsf}
	\DontPrintSemicolon 
	\SetKwInOut{Input}{Input}
	\SetKwInOut{Output}{Output}
	\SetKwRepeat{Do}{do}{while}
	\SetKwFunction{rangeSearch}{\textsc{booleanRangeSearch}}
	\SetKwFunction{ins}{\textsc{insert}}
        \SetKwFunction{lowerrange}{\textsc{SearchRange}}
        \SetKwFunction{invertrange}{\textsc{InvertRange}}
        \SetKwFunction{len}{\textsc{len}}
	\SetKwProg{myalg}{procedure}{}{}
	\SetKwData{tree}{\textsc{OrthogonalRangeSearchTree}}
	\SetKwData{op}{$\mathsf{op}$}
        \SetKwData{nonEqCount}{$\mathsf{k}$}
	\SetKwData{col}{$\mathsf{col}$}
	\Input{Relation $\R{R}$, Homogeneous DC $\varphi$}
        \Output{True/False}
	\SetKwProg{myproc}{\textsc{procedure}}{}{}
	\SetKwData{ret}{\textbf{return}}
        $\nonEqCount \leftarrow |\vars(\varphi) \setminus \vars_{=}(\varphi)|$ \;
        $H,T \leftarrow $ empty hash table \;
	\ForEach{$t \in R$ \label{line:for}}{
            $v \leftarrow \pi_{\vars_{=}(\varphi)} (t)$ \;
            \If{$v \not\in H$}{
                $H[v] \leftarrow $ new List()
            }
            $H[v].\ins(\pi_{\vars(\varphi) \setminus \vars_{=}(\varphi)}(t))$ \;
            \If{$\nonEqCount = 0 \land H[v].\len() > 1$}{
                \ret \textbf{false}        
            }
        }
        \ForEach{bucket $b \in H$}{
            $T[b] \leftarrow $ create \tree\ for all points in $H[b]$ \label{line:build}\;
            \ForEach{$t \in H[b]$}{
                $\bX, \bY \leftarrow \lowerrange{t}$ \tcc*{From Algorithm~\ref{algo:main}}
                \If{$T[b].\rangeSearch(\bX, \bY)$ \label{line:query}}{
                        \ret \textbf{false}
                }
            }
            free space used by $T[b]$
        }
        \ret \textbf{true}
	\caption{{\sc Batched DC verification}}
	\label{algo:main2}
\end{algorithm}}

Although the running time guarantees of \autoref{thm:main} and \autoref{thm:main2} only differ by a logarithmic factor in the worst-case, the performance on certain instances can vary significantly as shown below.

\begin{proposition}
    For every homogeneous DC $\varphi$ with at least one inequality predicate, there exists a relation $\R{R}$ such that Algorithm~\ref{algo:main} takes $O(1)$ time and Algorithm~\ref{algo:main2} requires $\Omega(|\R{R}|)$ time.
\end{proposition}
\begin{proof}
    We sketch the proof for $\varphi : \forall \tup{s}, \tup{t} \in \R{R},\ \neg (\tup{s}.A < \tup{t}.A)$ which can be extended straightforwardly for other DCs of interest. We construct a unary relation $\R{R}(A)$ of size $N$ as follows: the first tuple $t_1$ is $1$ and the remaining $N-1$ tuples are $2$. Note that $t_1$ forms a violation with every other tuple in the relation. Algorithm~\ref{algo:main} initializes one range tree for all the tuples since there is no equality predicate. After $t_1$ has been inserted, and the next tuple $t_2$ is inserted, the range search query will return true and the algorithm will terminate. Note that all the operations take $O(1)$ time since the tree only contains two tuples. However, Algorithm~\ref{algo:main2} requires $\Omega(|\R{R}|)$ time for the partitioning of the input in the for loop on~\autoref{line:for}.
\end{proof}

\smallskip
\noindent In other words, the proposition above says that \autoref{thm:main} has the advantage of terminating early and returning an answer since it builds the range tree incrementally and does the querying of the data structure on-the-fly. On the other hand, Algorithm~\ref{algo:main2} requires $\Omega(|\R{R}|)$ of processing before it can start querying the range trees. 

\smallskip
\introparagraph{Hybrid Verification (variant two)} We next present another interesting variant that combines the properties of the two algorithms that we have seen so far. The key idea is that we can sort the relation $\R{R}$ on one of the columns (say $C^\star$) that participates in a predicate with the inequality $<$ or $>$\footnote{The technique also works with $\leq$ and $\geq$ but requires a more involved construction which we omit due to space constraints.}. The comparator function used for sorting is $\op$ in the predicate. When creating the range tree, we can now ignore column $C^\star$, thus reducing the dimensions of the points inserted in the range tree by one. The details are shown in Algorithm~\ref{algo:main3}. We demonstrate with the help of an example.

\begin{algorithm}[!t]
	\SetCommentSty{textsf}
	\DontPrintSemicolon 
	\SetKwInOut{Input}{Input}
	\SetKwInOut{Output}{Output}
	\SetKwRepeat{Do}{do}{while}
	\SetKwFunction{rangeSearch}{\textsc{booleanRangeSearch}}
	\SetKwFunction{ins}{\textsc{insert}}
        \SetKwFunction{lowerrange}{\textsc{SearchRange}}
        \SetKwFunction{append}{\textsc{append}}
        \SetKwFunction{invertrange}{\textsc{InvertRange}}
        \SetKwFunction{len}{\textsc{len}}
	\SetKwProg{myalg}{procedure}{}{}
	\SetKwData{tree}{\textsc{OrthogonalRangeSearchTree}}
        \SetKwData{partition}{\textsc{HashPartition}}
	\SetKwData{op}{$\mathsf{op}$}
        \SetKwData{st}{$\mathsf{start}$}
        \SetKwData{ed}{$\mathsf{end}$}
        \SetKwData{nonEqCount}{$\mathsf{k}$}
	\SetKwData{col}{$\mathsf{col}$}
	\Input{Relation $\R{R}$, Homogeneous DC $\varphi$}
        \Output{True/False}
	\SetKwProg{myproc}{\textsc{procedure}}{}{}
	\SetKwData{ret}{\textbf{return}}
 \tcc{$C^\star$ participates in a predicate $p$ with $\op \in \{<,>\}$}
        $H \leftarrow \partition{\R{R}}$ \tcc{partition on $\vars_{=}(\varphi)$}
         $T \leftarrow $empty hash table \;
        \ForEach{bucket $b \in H$}{
            $\R{R}_b \leftarrow $ sorted list all tuples in $H[b]$ on column $C^\star$ \;
            $\mL \leftarrow \emptyset$ \;
            $T[b] \leftarrow \tree{}$ \;
            \For{$i \in \{0, \dots |\R{R}_b|-1\}$}{
                $\bX, \bY \leftarrow \lowerrange(t_i)$ \tcc*{From Algorithm~\ref{algo:main}}
                \If{$T[b].\rangeSearch(\bX, \bY)$}{
                        \ret \textbf{false}
                }
                \If{$t_i.C^\star < t_{i+1}.C^\star$}{
                    \If{$\mL \neq \emptyset$}{
                        \ForEach{$s \in \mL$}{
                            $H[v].\ins(\pi_{\vars(\varphi) \setminus \{\vars_{=}(\varphi),C^\star\}}(s))$
                        }
                        $\mL \leftarrow \emptyset$
                    } \Else{
                        $H[v].\ins(\pi_{\vars(\varphi) \setminus \{\vars_{=}(\varphi),C^\star\}}(t_i))$
                    }
                } \Else{
                    $\mL.\append(t_i)$
                }
            }
        }
        \ret \textbf{true}
	\caption{{\sc Hybrid DC verification}}
	\label{algo:main3}
\end{algorithm}

\begin{example}
    Let us look at the processing of $\phi: (\proj{s}{\mathtt{State}} = \proj{t}{\mathtt{State}} \land \proj{s}{\mathtt{Salary}} <  \proj{t}{\mathtt{Salary}} \land \proj{s}{\mathtt{FedTaxRate}} \geq  \proj{t}{\mathtt{FedTaxRate}})$ using range trees. First, we sort {\bfseries \texttt{Tax}} on \texttt{Salary} in ascending order and then remove the column. This step orders the tuples as $t_1, t_4, t_2, t_3$. This order guarantees that when we process a tuple $s$, it can form a violation only with tuples that appear before $s$. When $t_4$ is processed, the 1D tree is empty (note that $t_1$ belongs to a different partition). When $t_2$ is processed, the tree contains point $10$ (which is $t_4.\mathtt{FedTaxRate}$) and $t_2.\mathtt{FedTaxRate} = 15$ is larger than $10$. Since we are looking for a point that has a value of at least $10$ and there is none in the tree, we continue processing. However, if $t_4.\mathtt{FedTaxRate} = {\color{red} 22}$, then $t_2$ would form a violation with $t_2$ since $t_4$ has a higher salary but lower tax rate. Thus, the worst-case processing time for $\varphi_3$ will only be $O(|\R{R}| \log |\R{R}|)$ instead of $O(|\R{R}| \log^2 |\R{R}|)$.
\end{example}

Algorithm~\ref{algo:main3} has the desirable property of avoiding building the tree structure completely before we begin querying. Instead, it only requires sorting the dataset, which is cheaper. In many cases, the dataset may already be sorted on some column that is used in a DC predicate, thus avoiding the need to sort altogether. However, compared to Algorithm~\ref{algo:main}, it still requires a linear pass over the input. We conclude the sub-section by summarizing the main results obtained from this section in~\autoref{table:summary}.

\begin{table*}[!htp]
\caption{Summary of the main results} \label{table:summary}
\scalebox{0.9}{
 \begin{tabularx}{\linewidth}{ l  @{\extracolsep{\fill}} l l l} 
 \toprule
 Algorithm & Answering time & Space & Comments \\
 \midrule
Algorithm~\ref{algo:main} with range tree & $O(|\R{R}| \log^k |\R{R}|)$ & $O(|\R{R}| \log^{k-1} |\R{R}|)$  &  early termination\\

Algorithm~\ref{algo:main} with k-d tree & $O(|\R{R}|^{2-\frac{1}{k}} \log |\R{R}|)$ & $O(|\R{R}|)$ & early termination \\
\midrule
Algorithm~\ref{algo:main2} with range tree & $O(|\R{R}| \log^{k-1}|\R{R}|)$ & $O(|\R{R}| \log^{k-1} |\R{R}|)$  & batched version;  no early termination \\

Algorithm~\ref{algo:main2} with k-d tree & $O(|\R{R}|^{2-\frac{1}{k}})$ & $O(|\R{R}|)$  & batched version;  no early termination \\
\midrule
Algorithm~\ref{algo:main3} with range tree & $O(|\R{R}| (\log^{k-1}|\R{R}| + \log |\R{R}|))$ & $O(|\R{R}| \log^{k-1} |\R{R}|)$  & hybrid;  requires sorting \\

Algorithm~\ref{algo:main3} with k-d tree & $O(|\R{R}|^{2-\frac{1}{k-1}} \log |\R{R}|)$ & $O(|\R{R}|)$  & hybrid; $k \geq 2$;  requires sorting \\
\bottomrule
 \end{tabularx}
 }
 \end{table*}
}

\section{\sysg Discovery}
\label{sec:discovery}

\begin{algorithm}[t]
    \small
	\SetCommentSty{textsf}
	\DontPrintSemicolon 
	\SetKwInOut{Input}{Input}
	\SetKwInOut{Output}{Output}
	\SetKwRepeat{Do}{do}{while}
	\SetKwFunction{rangeSearch}{\textsc{booleanRangeSearch}}
	\SetKwFunction{ins}{\textsc{insert}}
        \SetKwFunction{lowerrange}{\textsc{SearchRange}}
        \SetKwFunction{len}{\textsc{len}}
        \SetKwFunction{new}{\textsc{new}}
        \SetKwFunction{append}{\textsc{append}}
        \SetKwFunction{List}{\textsc{list}}
	\SetKwProg{myalg}{procedure}{}{}
	\SetKwData{tree}{\textsc{OrthogonalRangeSearchTree}}
        \SetKwData{minimal}{\textsc{Minimal}}
        \SetKwData{nonredundant}{\textsc{NotPruned}}
        \SetKwData{verify}{\textsc{Verify}}
	\SetKwData{op}{$\mathsf{op}$}
        \SetKwData{nonEqCount}{$\mathsf{nonEqCount}$}
	\SetKwData{col}{$\mathsf{col}$}
        \SetKwData{ret}{\textbf{return}}
	\Input{Relation $\R{R}$}
        \Output{List of exact, minimal DCs}
	\SetKwProg{myproc}{\textsc{procedure}}{}{}
        $\mL \leftarrow \new\ \List()$ \tcc{stores exact, minimal DCs}
        $k \leftarrow 0$ \;
        \While{$k \leq |\vars(\R{R})|$}{
            $k \leftarrow k+1$ \;
            \ForEach{candidate $\varphi$ formed using size $k$ subset of $\vars(\R{R})$}{
            \tcc{\minimal\ borrowed from~\cite{chu2013discovering}}
                \If{$\minimal(\mL, \varphi)$ and $\nonredundant(\mL, \varphi)$ and $\verify(R, \varphi)$}{
                    output $\varphi$ \label{out} \;
                    $\mL.\append(\varphi)$
                }
            }
        }
        \ret\ processed $\mL$ using implication test \;from~\cite{chu2013discovering} 
        \myproc{$\nonredundant(\mL, \varphi)$}{
            \ForEach{$\varphi' \in \mL$}{
                $p_1, \dots, p_m \leftarrow \text{predicates in $\varphi'$}$ \;
                \ForEach{$j \in [m]$}{
                    \If{$\varphi$ contains $\{p_i\}_{i \neq j}$ and $\neg p_j$}
                    {
                        \ret \textbf{false}
                    }
                }
            }
            \ret \textbf{true}
        }
	\caption{{\sc DC discovery}}
	\label{algo:main4}
\end{algorithm}

\begin{figure*} 
    \includegraphics[width=\textwidth]{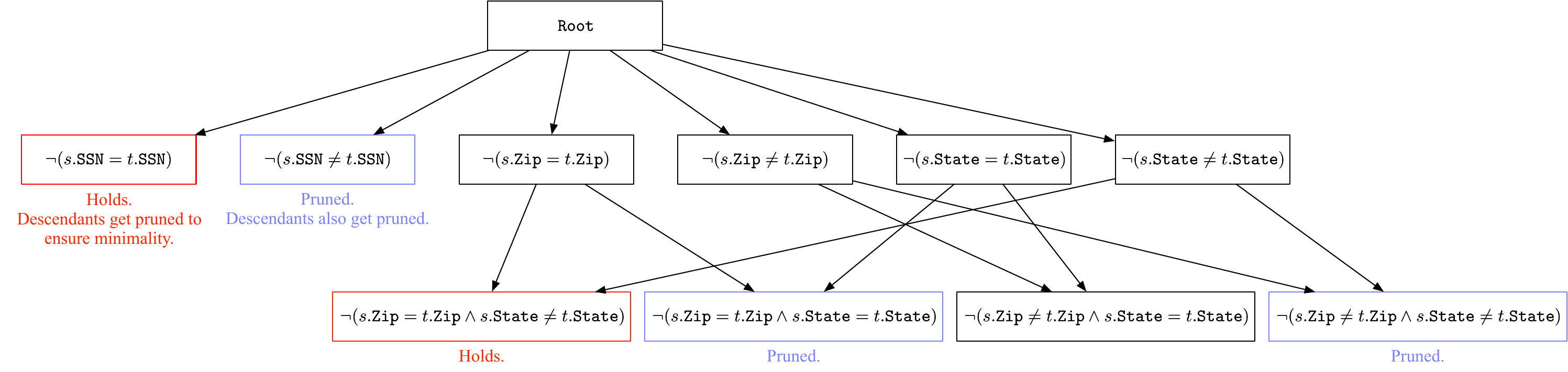}
    \caption{The search space of homogeneous DC discovery on Table~\ref{table:tax} when we only consider $\mathtt{SSN}$, $\mathtt{Zip}$, $\mathtt{State}$ and limit the number of predicates to two. DCs marked as ``Holds'' have been verified to be true and DCs marked as ``Pruned'' are candidates implied by DCs that are true.} \label{fig:latticeexample}
\end{figure*}

In this section, we propose a fast DC discovery algorithm. \eat{with desirable properties as outlined in Section~\ref{sec:intro}} To find all exact, minimal DCs, we use a lattice-based approach where we generate candidate DCs and leverage the verification algorithm to verify whether the DC holds.

Similar to prior works on functional dependency discovery~\cite{huhtala1999tane}, we start with singleton sets of attributes and traverse larger sets in a level-by-level fashion. For each set of attributes, we generate candidate DCs by generating all possible predicates. 
Then, we apply the verification algorithm to check whether the DC holds over the input. If the DC is true, we output the DC to the user and store it in a list $\mL$, which is used to check the minimality of a candidate DC. \eat{ To ensure that the DCs being verified are minimal, we use the stored list of verified DCs to ensure minimality.}
We also prune candidate DCs whose  validity is implied by other DCs. When $\neg (\bigwedge_{i \in [m]} p_i)$ is verified, for any $j \in [m]$, we remove all DCs containing $\{p_i\}_{i \neq j} \cup \{\neg p_j\}$ from the search space. 

\begin{example}
    Figure~\ref{fig:latticeexample} shows an example search sub-space showing the first two levels of the lattice. Level one (with incoming arrows from the \texttt{Root}) contains all DCs over a single column and level two contains candidates generated from predicates in level one. Consider the DC $\phi_1: \neg (s.\mathtt{SSN} = t.\mathtt{SSN})$ and suppose that it holds. Once this DC is verified, it is added to $\mL$ and does not contribute any new candidates in the search space. The next candidate $\phi_1: \neg (s.\mathtt{SSN} = t.\mathtt{SSN})$ is superfluous as it is guaranteed to be false due to logical implication. We also remove all descendants of $\phi_1$ because they will be equivalent to other DCs. For instance, $\neg(\proj{s}{\mathtt{SSN}} \neq \proj{t}{\mathtt{SSN}} \land \proj{s}{\mathtt{Zip}} = \proj{t}{\mathtt{Zip}})$ is equivalent to $\neg(\proj{s}{\mathtt{Zip}} = \proj{t}{\mathtt{Zip}})$, and the latter has already been checked on level one. On level two, the first candidate $\neg (\proj{s}{\mathtt{Zip}} = \proj{t}{\mathtt{Zip}} \land \proj{s}{\mathtt{State}} \neq \proj{t}{\mathtt{State}})$ holds and added to $\mL$, which helps us prune the two other candidates in the second level which are marked in the figure.\qed
\end{example}

\eat{
In the example presented in Figure~\ref{fig:latticeexample}, after $\neg(\proj{s}{\mathtt{SSN}} = \proj{t}{\mathtt{SSN}})$ is verified, there is no need to check $\neg(\proj{s}{\mathtt{SSN}} \neq \proj{t}{\mathtt{SSN}})$ since it must be false due to logical implication. We also remove all descendants of $\neg(\proj{s}{\mathtt{SSN}} \neq \proj{t}{\mathtt{SSN}})$ from consideration because they are equivalent to some other DCs. For instance, $\neg(\proj{s}{\mathtt{SSN}} \neq \proj{t}{\mathtt{SSN}} \land \proj{s}{\mathtt{Zip}} = \proj{t}{\mathtt{Zip}})$ is equivalent to $\neg(\proj{s}{\mathtt{Zip}} = \proj{t}{\mathtt{Zip}})$, and the latter has been checked in the previous level.}

\begin{table*}[!h] 
\caption{List of denial constraints used in experiments for each dataset.} \label{table:dc}
\begin{center}
 \resizebox{0.9\linewidth}{!}{\small \begin{tabularx}{\textwidth}{ c c c c l} 
 \toprule
 Dataset & Cardinality & \#Columns & DC number & Denial constraint \\
 \midrule
 \midrule
 $D_1$ & 50M & 28 & $\varphi_{1,1}$ &  $\neg (s.A = t.A \land s.B = t.B \land s.C \neq t.C \land s.D \neq t.D)$\\
 $D_1$ & 50M& 28 &$\varphi_{1,2}$ & $\neg (s.C = t.C \land s.E = t.E \land s.F = t.F \land s.G \neq t.G \land s.H \neq t.H)$ \\
$D_1$ & 50M& 28 &$\varphi_{1,3}$ & $\neg (s.B = t.B \land s.I = t.I \land s.J = t.J \land s.K \neq t.K \land s.L \neq t.L)$ \\
 $D_1$ & 50M& 28 &$\varphi_{1,4}$ &  $\neg (s.A = t.A \land s.I = t.I \land s.M > t.M \land s.N \neq t.N)$\\
 \midrule
 $D_2$ & 25M & 28&$\varphi_{2,1}$ &  $\neg (s.A=t.A \land s.B=t.B \land s.C \geq t.C \land s.D \leq t.D \land s.E \leq t.E \land s.F \geq t.F \land s.G>t.G)$ \\
 $D_2$ & 25M& 28&$\varphi_{2,2}$ &  $\neg (s.A \neq t.A \land s.B=t.B \land s.H \leq t.H \land s.F \geq t.F \land s.G \geq t.G)$ \\
 $D_2$ & 25M& 28&$\varphi_{2,3}$ &  $\neg (s.A=t.A \land s.I \neq t.I \land s.D \leq t.D \land s.G \geq t.G \land s.J=t.J)$ \\
 $D_2$ & 25M& 28&$\varphi_{2,4}$ &  $\neg (s.C \leq t.C \land s.D \leq t.D \land s.K=t.K)$ \\
 \midrule
 $D_3$ & 10M & 76 &$\varphi_{3,1}$ &  $\neg (s.A = t.A  \land  s.B \leq t.B  \land  s.C > t.C)$ \\
 $D_3$ & 10M& 76 &$\varphi_{3,2}$ &  $\neg (s.A = t.A  \land  s.B \leq t.B  \land  s.D > t.D)$ \\
 $D_3$ & 10M& 76 &$\varphi_{3,3}$ &  $\neg (s.A = t.A  \land  s.B < t.B  \land  s.E = t.E)$ \\
 $D_3$ & 10M& 76 &$\varphi_{3,4}$ &  $\neg (s.E = t.E  \land  s.F \neq t.F  \land  s.G \neq t.G)$ \\
  \midrule
 $D_4$ & 10M & 80&$\varphi_{4,1}$ &  $\neg (s.A=t.A \land s.B \neq t.B \land s.C=t.C \land s.D=t.D)$ \\
 $D_4$ & 10M& 80&$\varphi_{4,2}$ &  $\neg (s.E \neq t.E \land s.F=t.F \land s.G=t.G \land s.B \neq t.B)$ \\
 $D_4$ & 10M& 80&$\varphi_{4,3}$ &  $\neg (s.H=t.H \land s.I=t.I \land s.J \leq t.J \land s.K \geq t.K)$ \\
 $D_4$ & 10M& 80&$\varphi_{4,4}$ &  $\neg (s.L=t.L \land s.M \neq t.M \land s.N \neq t.N)$ \\
 \bottomrule
 \end{tabularx}}
\label{table:overhead}
\end{center}
\end{table*}

Algorithm~\ref{algo:main4} describes the steps for DC discovery. It is straightforward to see that the time complexity of the algorithm is the product of \# candidate DCs considered and the verification time complexity. Although the DCs in the output of Algorithm~\ref{algo:main4} are minimal and we prune candidates during the search, there may still exist exact, minimals DCs that can be implied by other exact, minimal DCs. As a post-processing step, we can use the implication test algorithm proposed by Chu et al.~\cite{chu2013discovering} to find as many such DCs as possible and remove them. This step is also followed by~\cite{pena2019discovery, pena2020efficient}. Note that the implication test does not guarantee removing all of the implied DCs, and the complete test is a coNP-complete problem~\cite{baudinet1999constraint}.

Note that our algorithm can be further improved by incorporating sampling-based verification as a pre-filter and collecting multiple DCs to verify and use the ideas from~\cite{pena2019discovery} to exploit common predicates between the candidates. Our solution is also \emph{embarrassingly parallel} and can be easily extended to use multiple processors for verifying candidates. However, we intentionally keep our implementation simple since it already works well in practice for our production customer use cases. We leave the important topic of incorporating these optimizations as a future study topic.

\smallskip
\looseness-1
\introparagraph{Comparison with prior work} First, observe that we have the capability to output the candidate DC to the user immediately after the verification has been done. Thus, the user can terminate the algorithm at any point in time, a desirable property since the user can interrupt the algorithm and still get answers. This \emph{anytime} property for DC discovery is a direct consequence of faster DC verification. Indeed, since all prior methods for DC verification may require a quadratic amount of space (and thus, time), they cannot be used for our setting. Second, since the lattice is traversed in increasing size of the candidate constraints, we get the added benefit of generating \emph{succinct} constraints, a desirable property~\cite{chu2013discovering} and in line with the minimum descriptor length principle. Lastly, our proposed solution is space efficient and requires only $O(|\R{R}| + |\mL|)$ amount of space when the verification method uses $k$-d trees,  allowing our solution to run on commodity machines. This property is not achievable by any other non-trivial discovery method known so far.
\section{Experimental Evaluation}
\label{sec:evaluation}
In this section, we report the results of our experimental evaluation. In particular, we seek to answer the following questions:
\begin{enumerate}[label=\textbf{(Q.\arabic*)}, ref=Q.\arabic*]
    \item What is the performance improvement (time and space) of the \sysg verification algorithm compared to \facet? \label{one}
    \item What is the impact of the optimizations proposed in Section~\ref{sec:opt} in the overall verification time? \label{two}
    \item How does the performance and scalability of \sysg discovery compare to existing solutions? \label{three}
\end{enumerate}

\subsection{Experimental Setting} \label{subsec:settings}
\eat{\begin{table}[!htp]
\caption{Cardinality (in millions) of cluster pairs constructed by \facet\ and number of nodes in the tree constructed by \syso\ and \syskd.} \label{table:space}
 \begin{tabularx}{\linewidth}{ c @{\extracolsep{\fill}} r  r  r} 
 \toprule
 DC number & \facet\ & \syso\ & \syskd\ \\
 \midrule
 $\varphi_{1,1}$ & 455 & 115 & 7.5 \\
 $\varphi_{1,2}$ & 605 & 160 & 10.8 \\
 $\varphi_{1,3}$ & 905 & 375 & 22.5 \\
 $\varphi_{1,4}$ & 1980 & 1170 & 50 \\
\midrule
 $\varphi_{2,1}$ & 1524 & 870 & 2.2 \\
 $\varphi_{2,2}$ & 1680 & 886 & 2.5 \\
 $\varphi_{2,3}$ & 781 & 420 & 1 \\
 $\varphi_{2,4}$ & 950 & 510 & 25 \\
\midrule
 $\varphi_{3,1}$ & 11 & 5.5 & 1.1 \\
 $\varphi_{3,2}$ & 11 & 7.4 & 1.8 \\
 $\varphi_{3,3}$ & 23 & 11 & 5.6 \\
 $\varphi_{3,4}$ & 55 & 20 & 10 \\
 \midrule
 $\varphi_{4,1}$ & 21 & 15 & 1.1 \\
 $\varphi_{4,2}$ & 45 & 15 & 1.9 \\
 $\varphi_{4,3}$ & 54 & 18 & 2.2 \\
 $\varphi_{4,4}$ & 98 & 55 & 10 \\
 \bottomrule
 \end{tabularx}
 \end{table}}
\begin{figure*}
    \scalebox{1}{
    \hspace*{-1em} \includegraphics[width=\textwidth]{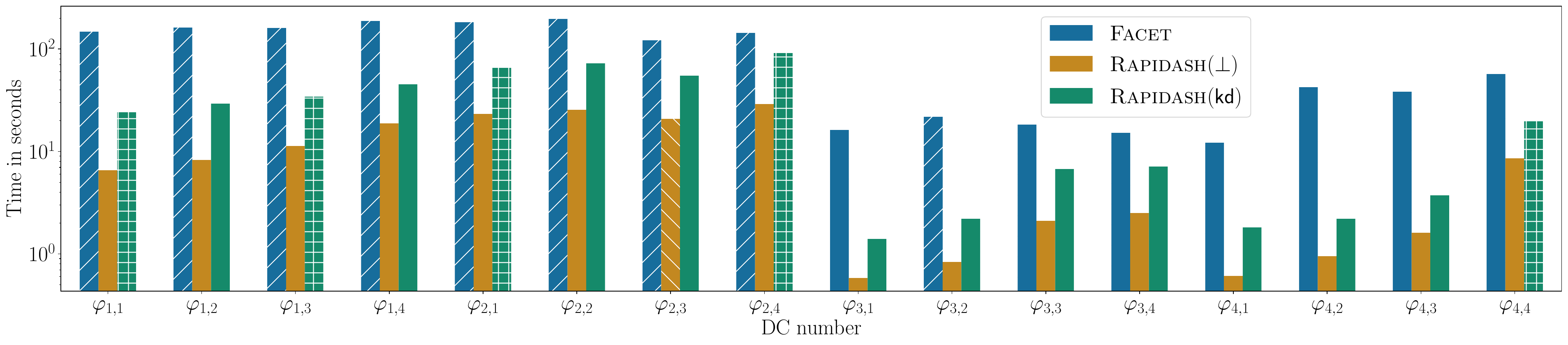}
    }
    \vspace{-2mm}
    \caption{Running time (in seconds) of different algorithms for DC verification.} \label{fig:runtime}
\end{figure*}

Table~\ref{table:dc} lists twelve DCs over four production datasets that we use in our experiments\footnote{The column names in the DCs have been omitted due to security and privacy concerns.}. Two of these datasets contain banking information, one dataset is related to the shipping of documents and products, and the last dataset contains sales information. Each dataset contains a mix of categorical, numeric, and datetime columns. 
To stress-test our algorithm and prior work, we use complex data quality rules that are either discovered automatically or manually verified to be meaningful. Further, to make sure that the DCs are not trivial to falsify, we pick 3 DCs by taking a $10\%$ sample of each dataset and discover DCs that are true over the sample. The fourth DC (denoted by $\varphi_{i, 4}$ for dataset $D_i$) holds over the full dataset. \eat{We do not use the DCs from~\cite{pena2021fast} since $7$ out of the $12$ constraints can be verified by our algorithm in provably linear time, and thus cannot be used to stress-test our proposed strategies.}

We ran all experiments on an Intel(R) Xeon(R) W-2255 CPU @ 3.70GHz machine with 128GB RAM running Windows 10 Enterprise (version 22H2). All of our experiments are executed over a single core and in the main memory setting. Similar to all prior work, we implemented our algorithm in Java. Despite our repeated attempts, we were not able to obtain the original \facet\ source code from the authors of~\cite{pena2022fast}. Therefore, we implemented \facet\ in Java using the Metanome infrastructure from \cite{bleifuss2017efficient} and \cite{pena2019discovery}, with all optimizations as described in~\cite{pena2022fast}. Further, for verification, we ensure that \facet\ execution terminates as soon as the first violation is found.
All reported running times are the trimmed mean of five independent executions after the dataset has been loaded in memory.

For verification (see Algorithm~\ref{algo:main}), we use a standard implementation of orthogonal range trees (referred to as \syso) and k-d trees (referred to as \syskd). Like \facet, \textsc{Rapidash} does not require any traditional column indexing. All data structures are built on-the-fly. \eat{We also implement the batched version of verification (Algorithm~\ref{algo:main2} with the improved optimization) to understand the tradeoffs between incremental construction vs. batched construction.} For discovery, we compare our proposed algorithm with \textsc{Hydra} and \textsc{DCFinder}, two state-of-the-art DC mining algorithms. We use the same predicate space as defined in~\cite{pena2020efficient}.

 \begin{figure*}
    \scalebox{0.95}{
    \hspace*{-1em} \includegraphics[width=\textwidth]{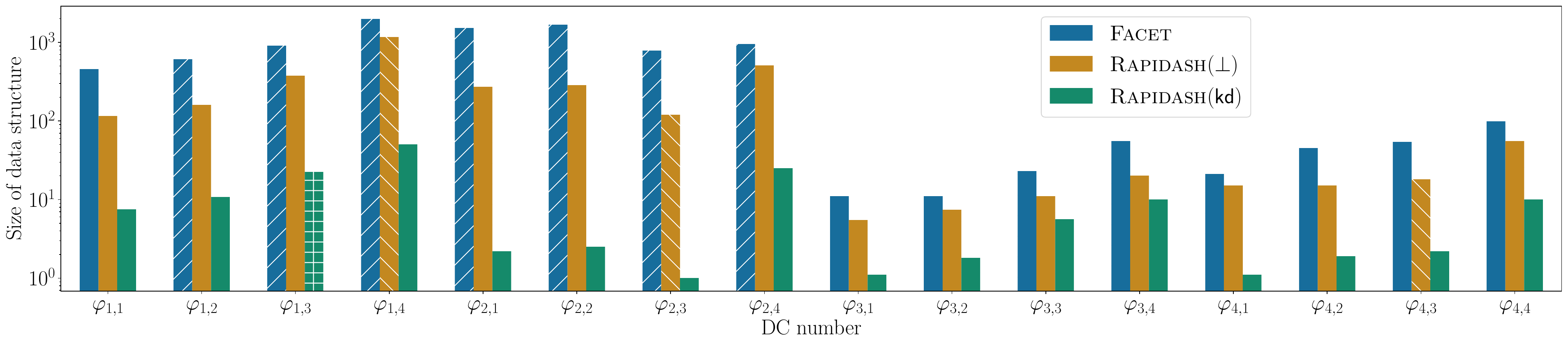}
    }
    \vspace{-2mm}
    \caption{Space requirement of different algorithms for DC verification. For \facet, space usage is the cardinality (in millions) of cluster pairs constructed and the number of nodes in the tree constructed for \syso\ and \syskd.} \label{fig:space}
\end{figure*}

\begin{figure*}
    \scalebox{1}{
    \includegraphics[width=1\textwidth]{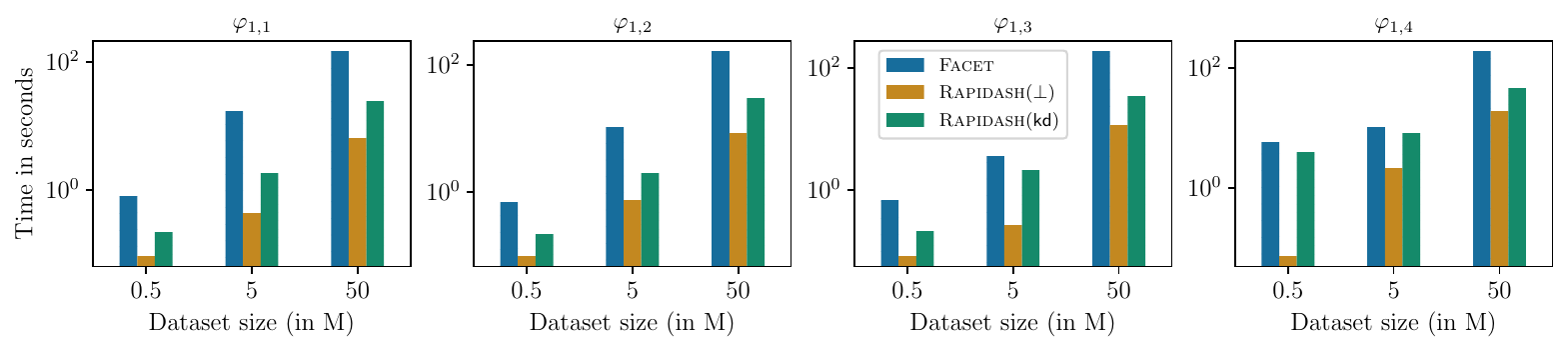}}
    \vspace{-4mm}
    \caption{Running time (in seconds) for DC verification on $D_1$ with varying cardinality.} \label{fig:scale}
    \vspace{-4mm}
\end{figure*}

\subsection{DC Verification} \label{subsec:comp}
In this section, we answer \textbf{\ref{one}} (performance) and \textbf{\ref{two}} (scalability and optimizations). Figure~\ref{fig:runtime} shows the running time (in log scale) of~\sys\ for all datasets and DCs. The speedup obtained by our algorithm is close to an order of magnitude and up to $40\times$. Compared to \facet, both algorithms perform significantly better on all DCs. \syso\ performs better than \syskd\ on all DCs. This is not surprising since using k-d trees for orthogonal range search is more expensive as shown by their big-O time analysis. However, \syskd\ is still faster than \facet\ by up to $20\times$. 

The speedup obtained by \sys\ can be attributed to two reasons. First, \sys\ can terminate as soon as a violation is discovered, as opposed to \facet\ that cannot do early termination in general (see Section~\ref{sec:limitation}). Therefore, in most cases, our algorithm does not require looking at the entire relation. This behavior can be empirically observed by noticing that the running time for $\varphi_{i,4}$ is greater than all other DCs for a fixed dataset. Since $\varphi_{i,4}$ holds for the entire dataset, the algorithm needs to process every row and thus there is no opportunity for early termination. 
$\varphi_{4,1}$ (functional dependency) also exemplifies the same observation.  Although \facet\ takes only $10$ seconds, \syso\ is over a magnitude faster.

The second reason is that \sys\ does not require any expensive materialization as opposed to the cluster pair generation that is done by \facet. As mentioned in Section~\ref{sec:limitation}, \facet\ constructs a refinement pipeline, and the intermediate representation of the cluster pairs can become expensive for predicates containing inequalities and disequalities. Figure~\ref{fig:space} shows the size of the data structures constructed by each algorithm. For \facet, the size is the cardinality of the cluster pairs generated at each stage of the refinement pipelines. For \sys\ based algorithms, the size refers to the number of points inserted in the tree. For all DCs, \facet\ used $1.4-8\times$ more space compared to \syso. \syskd\, was a further order of magnitude lower in its space requirement compared to \syso in line with the theoretical predictions presented in Section~\ref{sec:verification}. 

\smallskip
\introparagraph{Scalability} To study the scalability of \sys, we use dataset $D_1$ and vary the number of rows to understand the impact of input size on the running time of the DCs. Figure~\ref{fig:scale} shows the results when vary the dataset size of $D_1$ from $0.5$M to $50$M. Both \syso\ and \syskd\ scale almost linearly for the first three DCs. For $\varphi_{1,4}$, while \syso\ scales linearly, \syskd\ has super linear scalability, which is in line with the expectation. The behavior of \sys\ on other datasets was also very similar. The performance gap between \facet\ and our solution narrows when the dataset size is small. This is expected since \facet\ performance depends on the sizes of cluster pairs generated by refinements. If the size of the partitions generated after processing equality predicates is small, refinement processing of other non-equality predicates is not as expensive compared to larger partitions.

\smallskip
\introparagraph{Inequality Predicate Optimization} From the set of DCs considered, the single inequality optimization as described in Section~\ref{sec:opt} is applicable to $\varphi_{3,3}$ and $\varphi_{4,1}$\footnote{After converting disequality to inequality}. Using the homogeneous version of Algorithm~\ref{algo:one-inequality}, we observed a speedup of 1.2$\times$ and 1.1$\times$ respectively, compared to the fastest implementation \syso. The main reason for the improvement is that instead of creating a binary tree on the column in the inequality predicate, we only keep track of the minimum and maximum value of the column for the tuples present in a partition after hashing. 

\smallskip
\introparagraph{Disequality Predicate Optimization} The disequality predicate optimization is applicable to $\varphi_{1,1}, \varphi_{1,2}, \varphi_{1,3}$ from the four DCs on $D_1$. Each of the three constraints has exactly two disequality predicates. With the optimization switched on, we observed an improvement of $2\times$ for each of the three DCs owing to the fact that the optimization generates two candidates of four. Thus, the algorithm only needs to do half the work.

\eat{\smallskip
\introparagraph{Verification using Batched and Hybrid versions} The last experiment in this subsection compares the performance of constructing the entire tree up front. Table~\ref{table:speedup1} shows the main result. We observed that for \syso, using batched version is always slower. This is primarily due to no opportunity for early termination and range querying over the entire dataset instead of a partial subset. The same behavior was also observed for \syskd\ on $\varphi_{1,1}$ and $\varphi_{1,2}$. However, for $\varphi_{1,3}$ and $\varphi_{1,4}$, the batched version of \syskd\ provided speedups. This behavior is explained by the fact that creating kd-tree with the entire dataset allows for better balancing of the tree structure, which in turn improves the range search query performance~\cite{orenstein1984class}. We leave the deeper investigation of how to use improved versions of kd-trees (such as adaptive kd-trees) for faster performance as a problem for future work. \TODO{incomplete; need to add hybrid}

\begin{table}[!htp]
\caption{Speedup obtained using batched/improved batched over non-batched version for $D_1$.} \label{table:speedup1}
 \begin{tabularx}{\linewidth}{ l  @{\extracolsep{\fill}} r r r r} 
 \toprule
 Algorithm & $\varphi_{1,1}$ & $\varphi_{1,2}$ & $\varphi_{1,3}$& $\varphi_{1,4}$ \\
 \midrule
Batched \syso\ & 0.82 & 0.82 & 0.94 &  1.25 \\
Batched \syskd\ & 0.96 & 0.98 & 1.21 &  2.21 \\
\midrule
Hybrid \syso\ &  &  &  &   \\
Hybrid \syskd\ &  &  &  &  \\
\bottomrule
 \end{tabularx}
 \end{table}}

\subsection{DC Discovery}
\label{subsec:discovery}
This section is dedicated to answering \textbf{\ref{three}}.

\smallskip
\looseness-1
\introparagraph{Performance} We evaluate the performance of our DC discovery algorithm (\sysdisc) in comparison to \textsc{Hydra} and \textsc{DCFinder} using all four datasets. We'd like to note that our datasets are much larger than those used in prior work (\textsc{Hydra} and \textsc{DCFinder} were evaluated on datasets consisting of up to 1M rows). We run the experiments with a time limit of $48$ hours. For all datasets, both \textsc{Hydra} and \textsc{DCFinder} \textbf{{could not finish}} the computation of the evidence set within the time limit \emph{{for any dataset}}. This is not surprising since the evidence-set can be super linear in the size of the dataset and has an exponential dependency on the number of columns. Therefore, even after spending a lot of computing resources, the user does not get any information at all about whether there even exists a DC or not. In contrast, \sysdisc\ was able to discover all constraints over a single attribute (i.e. $k=1$) within $10$ minutes for all datasets. Constraints over single attributes are already interesting since it includes single-column primary keys, finding whether columns that are empty, or sorted in a particular order. All constraints are generated over pairs of attributes ($k=2$) within one hour of starting the discovery process. Further, since \sysdisc\ continuously outputs DCs, the user is able to still get useful information even if the algorithm isn't allowed to run to completion.

\smallskip
\introparagraph{Scalability micro-benchmark} To understand the scaling behavior, we create a micro-benchmark where we vary the number of rows and columns in dataset $D_1$ and run \textsc{DCFinder}\footnote{We omit a microbenchmark with \textsc{Hydra} since \textsc{DCFinder} is known to be faster~\cite{pena2019discovery}} to generate all constraints over at most three columns. We intentionally keep the dataset size small to ensure that \textsc{DCFinder} can actually terminate. Figure~\ref{fig:dcrow} shows the scalability wrt. to varying the cardinality of $D_1$. \textsc{DCFinder} is faster than our solution when the dataset size is $10^5$ rows but its running time grows very quickly as the dataset size increases, demonstrating the super-linear running time empirically. This growth is entirely due to the evidence-set computation step. \sysdisc\ on the other hand has a much slower growth in running time. The same behavior was also observed when varying the number of columns but keeping the cardinality as $5 \times 10^5$, as shown in Figure~\ref{fig:dccol}. Even for only $25$ columns in a small dataset, the evidence set computation becomes a blocker. Note that the jump in \textsc{DCFinder} running time when going from $10$ to $15$ columns is larger than when going from $15$ to $20$ columns. This is because the evidence set computation is sensitive to column cardinality and whether it is numeric or categorical. Recall that categorical columns only admit $=, \neq$ as operators but numerical columns can have any operator in the predicate. When numerical columns are added, not only do they generate $6$ row-level homogeneous predicates but also column-level and heterogenous predicates. This leads to a blowup in predicate space which in-turn makes the evidence set larger. 

\begin{figure}[!tp]
    \scalebox{0.92}{
    \includegraphics[width=\linewidth]{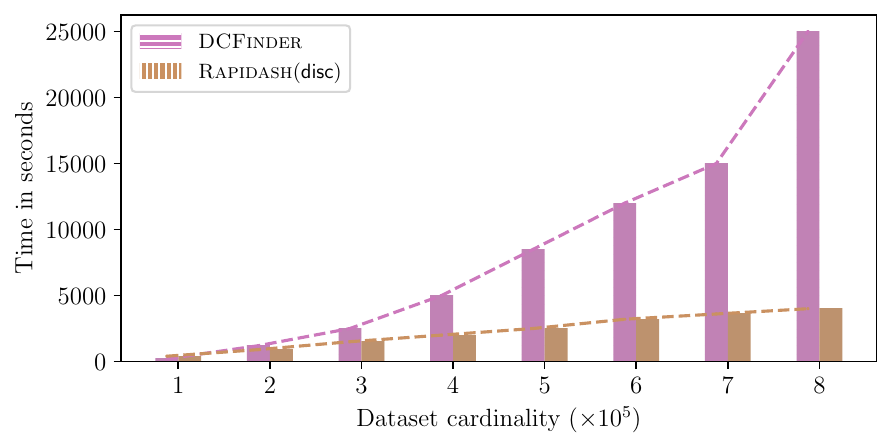}
    }
    \vspace{-3mm}
    \caption{Running time of \textsc{DCFinder} vs. \sysdisc\ for varying cardinality of $D_1$ with $15$ columns. \eat{\af{Use a pattern to make this plot accessible. Just distinguishing by color is not inclusive as color-blind people might not be able to distinguish it. We can make one solid color and the other dotted/dashed. This comment applies to all of the bar plots.}} }\label{fig:dcrow}
    \vspace{-4mm}
\end{figure}

\begin{figure}[t]
    \scalebox{0.92}{
    \includegraphics[width=\linewidth]{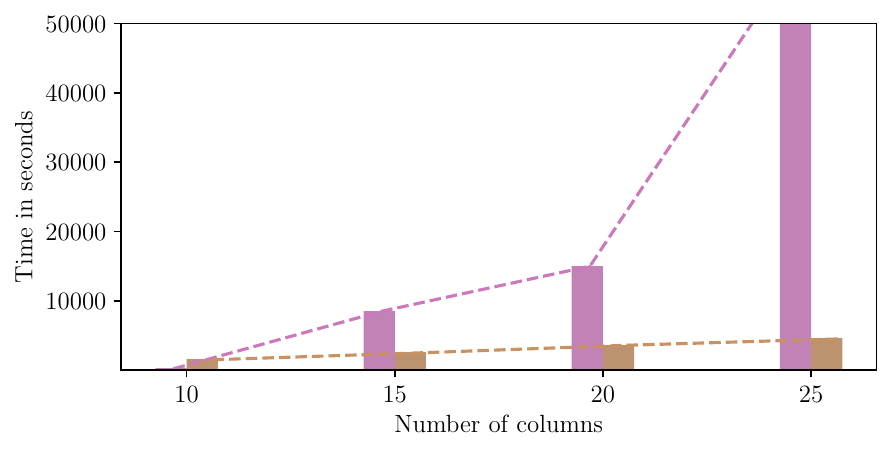}
    }
    \vspace{-2mm}
    \caption{Running time of \textsc{DCFinder} vs. \sysdisc\ for varying number of columns with $|D_1| = 5 \cdot 10^5$. Criss-cross hashed bar means experiment could not complete in $24$ hours.} \label{fig:dccol}
    \vspace{-4mm}
\end{figure}




\section{Related Work}
\label{sec:related}
DCs as an integrity constraint language was originally proposed by Chu et al.~\cite{chu2013discovering}. We refer the reader to~\cite{abedjan2017data, abedjan2018data} for a general overview.

\smallskip
\introparagraph{DC Verification} To the best of our knowledge, \textsc{Facet}~\cite{pena2021fast} is the state-of-the-art algorithm for DC verification. In more detail, given a DC, \textsc{Facet} is able to find all constraint violation pairs which is sufficient to detect whether a DC holds or not. \textsc{Facet} follows the design of \textsc{VioFinder}~\cite{pena2020efficient} and uses an operator
called \emph{refinement} to evaluate DC predicates. Our proposed algorithm has better worst-case time/space complexity than \textsc{Facet} which results in significant performance improvements in practice. Previous works on data cleaning~\cite{rekatsinas2017holoclean, geerts2020cleaning, fan2021parallel} rely on relational DBMSs to detect DC violations, where DCs are translated into SQL queries. Those DBMS-based methods often fall short when DCs contain inequalities, and they are slower than DC-specific methods by orders of magnitude as shown in ~\cite{pena2020efficient, pena2021fast}.

\introparagraph{DC Discovery} The two state-of-the-art systems that cater to \textbf{exact} DC discovery are \textsc{Hydra}~\cite{bleifuss2017efficient} and \textsc{DCFinder}~\cite{pena2019discovery}. Both of these systems rely on the two-step  process of first building the evidence set, followed by eumerating the DCs. In particular,~\cite{bleifuss2017efficient} proposed a hybrid strategy that combines DC discovery on a small sample with further refinement based on DC violations on the full instance. \textsc{DCFinder} is designed for approximate DC discovery, but its optimizations also apply to exact DC discovery. The two-step approach is also been successfully used for other dependency discovery algorithms~\cite{schirmer2020efficient, papenbrock2016hybrid}. Following \textsc{DCFinder}, several systems have been proposed for efficient \textbf{approximate} constraint discovery. \textsc{ADCMiner}~\cite{livshits2020approximate} is an extension to \textsc{DCFinder} that supports user-defined semantics of approximation and sampling is used to reduce the computation cost. \textsc{ECP}~\cite{pena2022fast} introduces customized data representations, indexes and algorithms for both evidence set building and DC enumeration to achieve better parallelism. \textsc{FastADC}~\cite{xiao2022fast} utilizes a condensed representation of evidence sets to support more efficient bit operations and cache utilization in the evidence set building stage and extends the evidence inversion technique from \textsc{Hydra} for approximate DC enumeration. All these algorithms are evidence-set based whose limitations have been discussed extensively in this paper. Our work focuses on exact DC discovery. Extending it to account for approximate discovery and compare with the techniques discussed above is left for future work. 

The lattice-based approach for restricted classes of constraint discovery (such as functional dependencies) has been employed by several works in the past~\cite{huhtala1999tane, abedjan2014dfd, papenbrock2016hybrid}. However, to the best of our knowledge, methods using lattice-based discovery have not been used for general DC mining as the validation of DCs using existing algorithms is expensive. Our work remedies this issue by proposing near-optimal algorithms for verifying any DC.

\introparagraph{Range Searching} The connection between geometric algorithms and general join query processing has been made by several prior works~\cite{khamis2016joins, ngo2014beyond}. Specifically, range searching has been used for aggregate query processing~\cite{khamis2020functional}. In fact, optimizations introduced in this paper could also be applied to certain queries considered in~\cite{wang2022conjunctive} since DCs can be expressed as CQs with comparisons. Range trees and their variants have also been extensively used in geospatial information systems (see~\cite{arge2008priority, guttman1984r, beckmann1990r, beckmann2009revised, cheung1998enhanced, kothuri2002quadtree, de2000computational} for an overview) and indexes for database systems~\cite{hellerstein1995generalized}. Ours is the first work to make the connection between constraint verification and discovery using ideas from computational geometry. For an overview of the theoretical aspects of range searching, we refer the reader to~\cite{agarwal2004range}.
\section{Conclusion and Future Work}
\label{sec:conclusions}
In this paper, we studied the problem of DC verification and discovery. We presented \sys, a DC verification algorithm with near-linear time complexity with respect to the dataset size that leverages prior work on orthogonal range search. We also developed an anytime DC discovery algorithm which does lattice search based on our verification algorithm. Unlike previous works, our discovery algorithm eliminates the reliance on the construction of evidence sets, which can be computationally expensive. Through empirical evaluation, we demonstrated that our DC verification algorithm is faster than the state of the art by an order of magnitude on large-scale production datasets. Our DC discovery algorithm is able to output valid DCs incrementally whereas existing methods fail to provide any useful information. This paper opens up a line of work that can leverage close connections between problems related to denial constraints and computational geometry. Our techniques can be extended to the approximate DC setting as well where we wish to ensure that the number of DC violations is at most some fraction of the dataset. 
Potential directions for future work include considering dynamic data and external memory setting, as well as further improving complexity~\cite{saxena2019distributed,fries1987log}.  
\eat{
We now outline some concrete directions for future work.


\introparagraph{DC discovery and verification over dynamic data} The connections to geometric data structures opens up an avenue to explore how to solve the problems considered in this paper on dynamic data. There is a rich history of dynamic algorithms for both range trees and $k$-d trees, and is likely to allow the development of elegant algorithms with provable guarantees.

\introparagraph{External memory setting} State-of-the art solutions for DC discovery in external memory also use evidence set as a primitive~\cite{saxena2019distributed}. It would be interesting to see how external memory variants of range search data structures compare to existing solutions.

\introparagraph{Further improvements in complexity} Although we are close to the optimal, some of the log factors can be further improved. In particular, for certain range queries (such as 3-sided ranges~\cite{fries1987log}), there exist better solutions compared to using range trees or $k$-d trees. Studying the integration of such solutions into our framework would likely result to more performance improvements.
\endeat}

\balance
\bibliographystyle{ACM-Reference-Format}
\bibliography{reference}


\begin{thebibliography}{46}


\ifx \showCODEN    \undefined \def \showCODEN     #1{\unskip}     \fi
\ifx \showDOI      \undefined \def \showDOI       #1{#1}\fi
\ifx \showISBNx    \undefined \def \showISBNx     #1{\unskip}     \fi
\ifx \showISBNxiii \undefined \def \showISBNxiii  #1{\unskip}     \fi
\ifx \showISSN     \undefined \def \showISSN      #1{\unskip}     \fi
\ifx \showLCCN     \undefined \def \showLCCN      #1{\unskip}     \fi
\ifx \shownote     \undefined \def \shownote      #1{#1}          \fi
\ifx \showarticletitle \undefined \def \showarticletitle #1{#1}   \fi
\ifx \showURL      \undefined \def \showURL       {\relax}        \fi
\providecommand\bibfield[2]{#2}
\providecommand\bibinfo[2]{#2}
\providecommand\natexlab[1]{#1}
\providecommand\showeprint[2][]{arXiv:#2}

\bibitem[Abedjan et~al\mbox{.}(2016)]%
        {abedjan2016detecting}
\bibfield{author}{\bibinfo{person}{Ziawasch Abedjan}, \bibinfo{person}{Xu Chu},
  \bibinfo{person}{Dong Deng}, \bibinfo{person}{Raul~Castro Fernandez},
  \bibinfo{person}{Ihab~F Ilyas}, \bibinfo{person}{Mourad Ouzzani},
  \bibinfo{person}{Paolo Papotti}, \bibinfo{person}{Michael Stonebraker}, {and}
  \bibinfo{person}{Nan Tang}.} \bibinfo{year}{2016}\natexlab{}.
\newblock \showarticletitle{Detecting data errors: Where are we and what needs
  to be done?}
\newblock \bibinfo{journal}{\emph{Proceedings of the VLDB Endowment}}
  \bibinfo{volume}{9}, \bibinfo{number}{12} (\bibinfo{year}{2016}),
  \bibinfo{pages}{993--1004}.
\newblock


\bibitem[Abedjan et~al\mbox{.}(2017)]%
        {abedjan2017data}
\bibfield{author}{\bibinfo{person}{Ziawasch Abedjan}, \bibinfo{person}{Lukasz
  Golab}, {and} \bibinfo{person}{Felix Naumann}.}
  \bibinfo{year}{2017}\natexlab{}.
\newblock \showarticletitle{Data profiling: A tutorial}. In
  \bibinfo{booktitle}{\emph{Proceedings of the 2017 ACM International
  Conference on Management of Data}}. \bibinfo{pages}{1747--1751}.
\newblock


\bibitem[Abedjan et~al\mbox{.}(2018)]%
        {abedjan2018data}
\bibfield{author}{\bibinfo{person}{Ziawasch Abedjan}, \bibinfo{person}{Lukasz
  Golab}, \bibinfo{person}{Felix Naumann}, {and} \bibinfo{person}{Thorsten
  Papenbrock}.} \bibinfo{year}{2018}\natexlab{}.
\newblock \showarticletitle{Data Profiling. Morgan \& Claypool Publishers}.
\newblock \bibinfo{journal}{\emph{Synthesis Lectures on Data Management}}
  (\bibinfo{year}{2018}).
\newblock


\bibitem[Abedjan et~al\mbox{.}(2014)]%
        {abedjan2014dfd}
\bibfield{author}{\bibinfo{person}{Ziawasch Abedjan}, \bibinfo{person}{Patrick
  Schulze}, {and} \bibinfo{person}{Felix Naumann}.}
  \bibinfo{year}{2014}\natexlab{}.
\newblock \showarticletitle{DFD: Efficient functional dependency discovery}. In
  \bibinfo{booktitle}{\emph{Proceedings of the 23rd ACM International
  Conference on Conference on Information and Knowledge Management}}.
  \bibinfo{pages}{949--958}.
\newblock


\bibitem[Agarwal(2004)]%
        {agarwal2004range}
\bibfield{author}{\bibinfo{person}{Pankaj~K Agarwal}.}
  \bibinfo{year}{2004}\natexlab{}.
\newblock \bibinfo{title}{Range Searching.}
\newblock
\newblock


\bibitem[Arge et~al\mbox{.}(2008)]%
        {arge2008priority}
\bibfield{author}{\bibinfo{person}{Lars Arge}, \bibinfo{person}{Mark~de Berg},
  \bibinfo{person}{Herman Haverkort}, {and} \bibinfo{person}{Ke Yi}.}
  \bibinfo{year}{2008}\natexlab{}.
\newblock \showarticletitle{The priority R-tree: A practically efficient and
  worst-case optimal R-tree}.
\newblock \bibinfo{journal}{\emph{ACM Transactions on Algorithms (TALG)}}
  \bibinfo{volume}{4}, \bibinfo{number}{1} (\bibinfo{year}{2008}),
  \bibinfo{pages}{1--30}.
\newblock


\bibitem[Baudinet et~al\mbox{.}(1999)]%
        {baudinet1999constraint}
\bibfield{author}{\bibinfo{person}{Marianne Baudinet}, \bibinfo{person}{Jan
  Chomicki}, {and} \bibinfo{person}{Pierre Wolper}.}
  \bibinfo{year}{1999}\natexlab{}.
\newblock \showarticletitle{Constraint-generating dependencies}.
\newblock \bibinfo{journal}{\emph{J. Comput. System Sci.}}
  \bibinfo{volume}{59}, \bibinfo{number}{1} (\bibinfo{year}{1999}),
  \bibinfo{pages}{94--115}.
\newblock


\bibitem[Beckmann et~al\mbox{.}(1990)]%
        {beckmann1990r}
\bibfield{author}{\bibinfo{person}{Norbert Beckmann},
  \bibinfo{person}{Hans-Peter Kriegel}, \bibinfo{person}{Ralf Schneider}, {and}
  \bibinfo{person}{Bernhard Seeger}.} \bibinfo{year}{1990}\natexlab{}.
\newblock \showarticletitle{The R*-tree: An efficient and robust access method
  for points and rectangles}. In \bibinfo{booktitle}{\emph{Proceedings of the
  1990 ACM SIGMOD international conference on Management of data}}.
  \bibinfo{pages}{322--331}.
\newblock


\bibitem[Beckmann and Seeger(2009)]%
        {beckmann2009revised}
\bibfield{author}{\bibinfo{person}{Norbert Beckmann} {and}
  \bibinfo{person}{Bernhard Seeger}.} \bibinfo{year}{2009}\natexlab{}.
\newblock \showarticletitle{A revised R*-tree in comparison with related index
  structures}. In \bibinfo{booktitle}{\emph{Proceedings of the 2009 ACM SIGMOD
  International Conference on Management of data}}. \bibinfo{pages}{799--812}.
\newblock


\bibitem[Bentley and Friedman(1979)]%
        {bentley1979data}
\bibfield{author}{\bibinfo{person}{Jon~Louis Bentley} {and}
  \bibinfo{person}{Jerome~H Friedman}.} \bibinfo{year}{1979}\natexlab{}.
\newblock \showarticletitle{Data structures for range searching}.
\newblock \bibinfo{journal}{\emph{ACM Computing Surveys (CSUR)}}
  \bibinfo{volume}{11}, \bibinfo{number}{4} (\bibinfo{year}{1979}),
  \bibinfo{pages}{397--409}.
\newblock


\bibitem[Bentley and Saxe(1980)]%
        {bentley1980decomposable}
\bibfield{author}{\bibinfo{person}{Jon~Louis Bentley} {and}
  \bibinfo{person}{James~B Saxe}.} \bibinfo{year}{1980}\natexlab{}.
\newblock \showarticletitle{Decomposable searching problems I.
  Static-to-dynamic transformation}.
\newblock \bibinfo{journal}{\emph{Journal of Algorithms}} \bibinfo{volume}{1},
  \bibinfo{number}{4} (\bibinfo{year}{1980}), \bibinfo{pages}{301--358}.
\newblock


\bibitem[Bleifu{\ss} et~al\mbox{.}(2017)]%
        {bleifuss2017efficient}
\bibfield{author}{\bibinfo{person}{Tobias Bleifu{\ss}},
  \bibinfo{person}{Sebastian Kruse}, {and} \bibinfo{person}{Felix Naumann}.}
  \bibinfo{year}{2017}\natexlab{}.
\newblock \showarticletitle{Efficient denial constraint discovery with hydra}.
\newblock \bibinfo{journal}{\emph{Proceedings of the VLDB Endowment}}
  \bibinfo{volume}{11}, \bibinfo{number}{3} (\bibinfo{year}{2017}),
  \bibinfo{pages}{311--323}.
\newblock


\bibitem[Cheung and Fu(1998)]%
        {cheung1998enhanced}
\bibfield{author}{\bibinfo{person}{King~Lum Cheung} {and} \bibinfo{person}{Ada
  Wai-Chee Fu}.} \bibinfo{year}{1998}\natexlab{}.
\newblock \showarticletitle{Enhanced nearest neighbour search on the R-tree}.
\newblock \bibinfo{journal}{\emph{ACM SIGMOD Record}} \bibinfo{volume}{27},
  \bibinfo{number}{3} (\bibinfo{year}{1998}), \bibinfo{pages}{16--21}.
\newblock


\bibitem[Chu et~al\mbox{.}(2013)]%
        {chu2013discovering}
\bibfield{author}{\bibinfo{person}{Xu Chu}, \bibinfo{person}{Ihab~F Ilyas},
  {and} \bibinfo{person}{Paolo Papotti}.} \bibinfo{year}{2013}\natexlab{}.
\newblock \showarticletitle{Discovering denial constraints}.
\newblock \bibinfo{journal}{\emph{Proceedings of the VLDB Endowment}}
  \bibinfo{volume}{6}, \bibinfo{number}{13} (\bibinfo{year}{2013}),
  \bibinfo{pages}{1498--1509}.
\newblock


\bibitem[De~Berg(2000)]%
        {de2000computational}
\bibfield{author}{\bibinfo{person}{Mark De~Berg}.}
  \bibinfo{year}{2000}\natexlab{}.
\newblock \bibinfo{booktitle}{\emph{Computational geometry: algorithms and
  applications}}.
\newblock \bibinfo{publisher}{Springer Science \& Business Media}.
\newblock


\bibitem[Fan et~al\mbox{.}(2021)]%
        {fan2021parallel}
\bibfield{author}{\bibinfo{person}{Wenfei Fan}, \bibinfo{person}{Chao Tian},
  \bibinfo{person}{Yanghao Wang}, {and} \bibinfo{person}{Qiang Yin}.}
  \bibinfo{year}{2021}\natexlab{}.
\newblock \showarticletitle{Parallel discrepancy detection and incremental
  detection}.
\newblock \bibinfo{journal}{\emph{Proceedings of the VLDB Endowment}}
  \bibinfo{volume}{14}, \bibinfo{number}{8} (\bibinfo{year}{2021}),
  \bibinfo{pages}{1351--1364}.
\newblock


\bibitem[Fariha et~al\mbox{.}(2021)]%
        {Fariha0RGM21}
\bibfield{author}{\bibinfo{person}{Anna Fariha}, \bibinfo{person}{Ashish
  Tiwari}, \bibinfo{person}{Arjun Radhakrishna}, \bibinfo{person}{Sumit
  Gulwani}, {and} \bibinfo{person}{Alexandra Meliou}.}
  \bibinfo{year}{2021}\natexlab{}.
\newblock \showarticletitle{Conformance Constraint Discovery: Measuring Trust
  in Data-Driven Systems}. In \bibinfo{booktitle}{\emph{{SIGMOD} '21:
  International Conference on Management of Data, Virtual Event, China, June
  20-25, 2021}}, \bibfield{editor}{\bibinfo{person}{Guoliang Li},
  \bibinfo{person}{Zhanhuai Li}, \bibinfo{person}{Stratos Idreos}, {and}
  \bibinfo{person}{Divesh Srivastava}} (Eds.). \bibinfo{publisher}{{ACM}},
  \bibinfo{pages}{499--512}.
\newblock
\urldef\tempurl%
\url{https://doi.org/10.1145/3448016.3452795}
\showDOI{\tempurl}


\bibitem[Fries et~al\mbox{.}(1987)]%
        {fries1987log}
\bibfield{author}{\bibinfo{person}{Otfried Fries}, \bibinfo{person}{Kurt
  Mehlhorn}, \bibinfo{person}{Stefan N{\"a}her}, {and}
  \bibinfo{person}{Athanasios Tsakalidis}.} \bibinfo{year}{1987}\natexlab{}.
\newblock \showarticletitle{A log log n data structure for three-sided range
  queries}.
\newblock \bibinfo{journal}{\emph{Inform. Process. Lett.}}
  \bibinfo{volume}{25}, \bibinfo{number}{4} (\bibinfo{year}{1987}),
  \bibinfo{pages}{269--273}.
\newblock


\bibitem[Ge et~al\mbox{.}(2021)]%
        {gekamino}
\bibfield{author}{\bibinfo{person}{Chang Ge}, \bibinfo{person}{Shubhankar
  Mohapatra}, \bibinfo{person}{Xi He}, {and} \bibinfo{person}{Ihab~F. Ilyas}.}
  \bibinfo{year}{2021}\natexlab{}.
\newblock \showarticletitle{Kamino: Constraint-Aware Differentially Private
  Data Synthesis}.
\newblock \bibinfo{journal}{\emph{Proc. {VLDB} Endow.}} \bibinfo{volume}{14},
  \bibinfo{number}{10} (\bibinfo{year}{2021}), \bibinfo{pages}{1886--1899}.
\newblock
\urldef\tempurl%
\url{https://doi.org/10.14778/3467861.3467876}
\showDOI{\tempurl}


\bibitem[Geerts et~al\mbox{.}(2020)]%
        {geerts2020cleaning}
\bibfield{author}{\bibinfo{person}{Floris Geerts},
  \bibinfo{person}{Giansalvatore Mecca}, \bibinfo{person}{Paolo Papotti}, {and}
  \bibinfo{person}{Donatello Santoro}.} \bibinfo{year}{2020}\natexlab{}.
\newblock \showarticletitle{Cleaning data with Llunatic}.
\newblock \bibinfo{journal}{\emph{The VLDB Journal}}  \bibinfo{volume}{29}
  (\bibinfo{year}{2020}), \bibinfo{pages}{867--892}.
\newblock


\bibitem[Giannakopoulou et~al\mbox{.}(2020)]%
        {giannakopoulou2020cleaning}
\bibfield{author}{\bibinfo{person}{Stella Giannakopoulou},
  \bibinfo{person}{Manos Karpathiotakis}, {and} \bibinfo{person}{Anastasia
  Ailamaki}.} \bibinfo{year}{2020}\natexlab{}.
\newblock \showarticletitle{Cleaning denial constraint violations through
  relaxation}. In \bibinfo{booktitle}{\emph{Proceedings of the 2020 ACM SIGMOD
  International Conference on Management of Data}}. \bibinfo{pages}{805--815}.
\newblock


\bibitem[GryA(2012)]%
        {grya2012fundamentaab}
\bibfield{author}{\bibinfo{person}{Jaroslaw SAlBCDEaF Parke GodfreyF~Jarek
  GryA}.} \bibinfo{year}{2012}\natexlab{}.
\newblock \showarticletitle{Fundamentals of Ordering Dependencies}.
\newblock \bibinfo{journal}{\emph{Proceedings of the VLDB Endowment}}
  \bibinfo{volume}{5}, \bibinfo{number}{11} (\bibinfo{year}{2012}).
\newblock


\bibitem[Guttman(1984)]%
        {guttman1984r}
\bibfield{author}{\bibinfo{person}{Antonin Guttman}.}
  \bibinfo{year}{1984}\natexlab{}.
\newblock \showarticletitle{R-trees: A dynamic index structure for spatial
  searching}. In \bibinfo{booktitle}{\emph{Proceedings of the 1984 ACM SIGMOD
  international conference on Management of data}}. \bibinfo{pages}{47--57}.
\newblock


\bibitem[Hellerstein et~al\mbox{.}(1995)]%
        {hellerstein1995generalized}
\bibfield{author}{\bibinfo{person}{Joseph~M Hellerstein},
  \bibinfo{person}{Jeffrey~F Naughton}, {and} \bibinfo{person}{Avi Pfeffer}.}
  \bibinfo{year}{1995}\natexlab{}.
\newblock \bibinfo{booktitle}{\emph{Generalized search trees for database
  systems}}.
\newblock \bibinfo{publisher}{September}.
\newblock


\bibitem[Hopcroft et~al\mbox{.}(2001)]%
        {hopcroft2001introduction}
\bibfield{author}{\bibinfo{person}{John~E Hopcroft}, \bibinfo{person}{Rajeev
  Motwani}, {and} \bibinfo{person}{Jeffrey~D Ullman}.}
  \bibinfo{year}{2001}\natexlab{}.
\newblock \showarticletitle{Introduction to automata theory, languages, and
  computation}.
\newblock \bibinfo{journal}{\emph{Acm Sigact News}} \bibinfo{volume}{32},
  \bibinfo{number}{1} (\bibinfo{year}{2001}), \bibinfo{pages}{60--65}.
\newblock


\bibitem[Huhtala et~al\mbox{.}(1999)]%
        {huhtala1999tane}
\bibfield{author}{\bibinfo{person}{Yka Huhtala}, \bibinfo{person}{Juha
  K{\"a}rkk{\"a}inen}, \bibinfo{person}{Pasi Porkka}, {and}
  \bibinfo{person}{Hannu Toivonen}.} \bibinfo{year}{1999}\natexlab{}.
\newblock \showarticletitle{TANE: An efficient algorithm for discovering
  functional and approximate dependencies}.
\newblock \bibinfo{journal}{\emph{The computer journal}} \bibinfo{volume}{42},
  \bibinfo{number}{2} (\bibinfo{year}{1999}), \bibinfo{pages}{100--111}.
\newblock


\bibitem[Ibaraki et~al\mbox{.}(1999)]%
        {ibaraki1999functional}
\bibfield{author}{\bibinfo{person}{Toshihide Ibaraki},
  \bibinfo{person}{Alexander Kogan}, {and} \bibinfo{person}{Kazuhisa Makino}.}
  \bibinfo{year}{1999}\natexlab{}.
\newblock \showarticletitle{Functional dependencies in Horn theories}.
\newblock \bibinfo{journal}{\emph{Artificial Intelligence}}
  \bibinfo{volume}{108}, \bibinfo{number}{1-2} (\bibinfo{year}{1999}),
  \bibinfo{pages}{1--30}.
\newblock


\bibitem[Khamis et~al\mbox{.}(2020)]%
        {khamis2020functional}
\bibfield{author}{\bibinfo{person}{Mahmoud~Abo Khamis}, \bibinfo{person}{Ryan~R
  Curtin}, \bibinfo{person}{Benjamin Moseley}, \bibinfo{person}{Hung~Q Ngo},
  \bibinfo{person}{XuanLong Nguyen}, \bibinfo{person}{Dan Olteanu}, {and}
  \bibinfo{person}{Maximilian Schleich}.} \bibinfo{year}{2020}\natexlab{}.
\newblock \showarticletitle{Functional aggregate queries with additive
  inequalities}.
\newblock \bibinfo{journal}{\emph{ACM Transactions on Database Systems (TODS)}}
  \bibinfo{volume}{45}, \bibinfo{number}{4} (\bibinfo{year}{2020}),
  \bibinfo{pages}{1--41}.
\newblock


\bibitem[Khamis et~al\mbox{.}(2016)]%
        {khamis2016joins}
\bibfield{author}{\bibinfo{person}{Mahmoud~Abo Khamis}, \bibinfo{person}{Hung~Q
  Ngo}, \bibinfo{person}{Christopher R{\'e}}, {and} \bibinfo{person}{Atri
  Rudra}.} \bibinfo{year}{2016}\natexlab{}.
\newblock \showarticletitle{Joins via geometric resolutions: Worst case and
  beyond}.
\newblock \bibinfo{journal}{\emph{ACM Transactions on Database Systems (TODS)}}
  \bibinfo{volume}{41}, \bibinfo{number}{4} (\bibinfo{year}{2016}),
  \bibinfo{pages}{1--45}.
\newblock


\bibitem[Khayyat et~al\mbox{.}(2015)]%
        {khayyat2015lightning}
\bibfield{author}{\bibinfo{person}{Zuhair Khayyat}, \bibinfo{person}{William
  Lucia}, \bibinfo{person}{Meghna Singh}, \bibinfo{person}{Mourad Ouzzani},
  \bibinfo{person}{Paolo Papotti}, \bibinfo{person}{Jorge-Arnulfo
  Quian{\'e}-Ruiz}, \bibinfo{person}{Nan Tang}, {and} \bibinfo{person}{Panos
  Kalnis}.} \bibinfo{year}{2015}\natexlab{}.
\newblock \showarticletitle{Lightning fast and space efficient inequality
  joins}.
\newblock  (\bibinfo{year}{2015}).
\newblock


\bibitem[Kossmann et~al\mbox{.}(2022)]%
        {kossmann2022data}
\bibfield{author}{\bibinfo{person}{Jan Kossmann}, \bibinfo{person}{Thorsten
  Papenbrock}, {and} \bibinfo{person}{Felix Naumann}.}
  \bibinfo{year}{2022}\natexlab{}.
\newblock \showarticletitle{Data dependencies for query optimization: a
  survey}.
\newblock \bibinfo{journal}{\emph{The VLDB Journal}} \bibinfo{volume}{31},
  \bibinfo{number}{1} (\bibinfo{year}{2022}), \bibinfo{pages}{1--22}.
\newblock


\bibitem[Kothuri et~al\mbox{.}(2002)]%
        {kothuri2002quadtree}
\bibfield{author}{\bibinfo{person}{Ravi Kanth~V Kothuri}, \bibinfo{person}{Siva
  Ravada}, {and} \bibinfo{person}{Daniel Abugov}.}
  \bibinfo{year}{2002}\natexlab{}.
\newblock \showarticletitle{Quadtree and R-tree indexes in oracle spatial: a
  comparison using GIS data}. In \bibinfo{booktitle}{\emph{Proceedings of the
  2002 ACM SIGMOD international conference on Management of data}}.
  \bibinfo{pages}{546--557}.
\newblock


\bibitem[Livshits et~al\mbox{.}(2020)]%
        {livshits2020approximate}
\bibfield{author}{\bibinfo{person}{Ester Livshits}, \bibinfo{person}{Alireza
  Heidari}, \bibinfo{person}{Ihab~F Ilyas}, {and} \bibinfo{person}{Benny
  Kimelfeld}.} \bibinfo{year}{2020}\natexlab{}.
\newblock \showarticletitle{Approximate denial constraints}.
\newblock \bibinfo{journal}{\emph{arXiv preprint arXiv:2005.08540}}
  (\bibinfo{year}{2020}).
\newblock


\bibitem[Ngo et~al\mbox{.}(2014)]%
        {ngo2014beyond}
\bibfield{author}{\bibinfo{person}{Hung~Q Ngo}, \bibinfo{person}{Dung~T
  Nguyen}, \bibinfo{person}{Christopher Re}, {and} \bibinfo{person}{Atri
  Rudra}.} \bibinfo{year}{2014}\natexlab{}.
\newblock \showarticletitle{Beyond worst-case analysis for joins with
  minesweeper}. In \bibinfo{booktitle}{\emph{Proceedings of the 33rd ACM
  SIGMOD-SIGACT-SIGART symposium on Principles of database systems}}.
  \bibinfo{pages}{234--245}.
\newblock


\bibitem[Overmars(1983)]%
        {overmars1983design}
\bibfield{author}{\bibinfo{person}{Mark~H Overmars}.}
  \bibinfo{year}{1983}\natexlab{}.
\newblock \bibinfo{booktitle}{\emph{The design of dynamic data structures}}.
  Vol.~\bibinfo{volume}{156}.
\newblock \bibinfo{publisher}{Springer Science \& Business Media}.
\newblock


\bibitem[Papenbrock and Naumann(2016)]%
        {papenbrock2016hybrid}
\bibfield{author}{\bibinfo{person}{Thorsten Papenbrock} {and}
  \bibinfo{person}{Felix Naumann}.} \bibinfo{year}{2016}\natexlab{}.
\newblock \showarticletitle{A hybrid approach to functional dependency
  discovery}. In \bibinfo{booktitle}{\emph{Proceedings of the 2016
  International Conference on Management of Data}}. \bibinfo{pages}{821--833}.
\newblock


\bibitem[Pena et~al\mbox{.}(2019)]%
        {pena2019discovery}
\bibfield{author}{\bibinfo{person}{Eduardo~HM Pena}, \bibinfo{person}{Eduardo~C
  de Almeida}, {and} \bibinfo{person}{Felix Naumann}.}
  \bibinfo{year}{2019}\natexlab{}.
\newblock \showarticletitle{Discovery of approximate (and exact) denial
  constraints}.
\newblock \bibinfo{journal}{\emph{Proceedings of the VLDB Endowment}}
  \bibinfo{volume}{13}, \bibinfo{number}{3} (\bibinfo{year}{2019}),
  \bibinfo{pages}{266--278}.
\newblock


\bibitem[Pena et~al\mbox{.}(2021)]%
        {pena2021fast}
\bibfield{author}{\bibinfo{person}{Eduardo~HM Pena}, \bibinfo{person}{Eduardo~C
  de Almeida}, {and} \bibinfo{person}{Felix Naumann}.}
  \bibinfo{year}{2021}\natexlab{}.
\newblock \showarticletitle{Fast detection of denial constraint violations}.
\newblock \bibinfo{journal}{\emph{Proceedings of the VLDB Endowment}}
  \bibinfo{volume}{15}, \bibinfo{number}{4} (\bibinfo{year}{2021}),
  \bibinfo{pages}{859--871}.
\newblock


\bibitem[Pena et~al\mbox{.}(2020)]%
        {pena2020efficient}
\bibfield{author}{\bibinfo{person}{Eduardo~HM Pena}, \bibinfo{person}{Edson~R
  Lucas~Filho}, \bibinfo{person}{Eduardo~C de Almeida}, {and}
  \bibinfo{person}{Felix Naumann}.} \bibinfo{year}{2020}\natexlab{}.
\newblock \showarticletitle{Efficient detection of data dependency violations}.
  In \bibinfo{booktitle}{\emph{Proceedings of the 29th ACM International
  Conference on Information \& Knowledge Management}}.
  \bibinfo{pages}{1235--1244}.
\newblock


\bibitem[Pena et~al\mbox{.}(2022)]%
        {pena2022fast}
\bibfield{author}{\bibinfo{person}{Eduardo~HM Pena}, \bibinfo{person}{Fabio
  Porto}, {and} \bibinfo{person}{Felix Naumann}.}
  \bibinfo{year}{2022}\natexlab{}.
\newblock \showarticletitle{Fast Algorithms for Denial Constraint Discovery}.
\newblock \bibinfo{journal}{\emph{Proceedings of the VLDB Endowment}}
  \bibinfo{volume}{16}, \bibinfo{number}{4} (\bibinfo{year}{2022}),
  \bibinfo{pages}{684--696}.
\newblock


\bibitem[Rekatsinas et~al\mbox{.}(2017)]%
        {rekatsinas2017holoclean}
\bibfield{author}{\bibinfo{person}{Theodoros Rekatsinas}, \bibinfo{person}{Xu
  Chu}, \bibinfo{person}{Ihab~F Ilyas}, {and} \bibinfo{person}{Christopher
  R{\'e}}.} \bibinfo{year}{2017}\natexlab{}.
\newblock \showarticletitle{Holoclean: Holistic data repairs with probabilistic
  inference}.
\newblock \bibinfo{journal}{\emph{arXiv preprint arXiv:1702.00820}}
  (\bibinfo{year}{2017}).
\newblock


\bibitem[Saxena et~al\mbox{.}(2019)]%
        {saxena2019distributed}
\bibfield{author}{\bibinfo{person}{Hemant Saxena}, \bibinfo{person}{Lukasz
  Golab}, {and} \bibinfo{person}{Ihab~F Ilyas}.}
  \bibinfo{year}{2019}\natexlab{}.
\newblock \showarticletitle{Distributed implementations of dependency discovery
  algorithms}.
\newblock \bibinfo{journal}{\emph{Proceedings of the VLDB Endowment}}
  \bibinfo{volume}{12}, \bibinfo{number}{11} (\bibinfo{year}{2019}),
  \bibinfo{pages}{1624--1636}.
\newblock


\bibitem[Schirmer et~al\mbox{.}(2020)]%
        {schirmer2020efficient}
\bibfield{author}{\bibinfo{person}{Philipp Schirmer}, \bibinfo{person}{Thorsten
  Papenbrock}, \bibinfo{person}{Ioannis Koumarelas}, {and}
  \bibinfo{person}{Felix Naumann}.} \bibinfo{year}{2020}\natexlab{}.
\newblock \showarticletitle{Efficient discovery of matching dependencies}.
\newblock \bibinfo{journal}{\emph{ACM Transactions on Database Systems (TODS)}}
  \bibinfo{volume}{45}, \bibinfo{number}{3} (\bibinfo{year}{2020}),
  \bibinfo{pages}{1--33}.
\newblock


\bibitem[Wang and Yi(2022)]%
        {wang2022conjunctive}
\bibfield{author}{\bibinfo{person}{Qichen Wang} {and} \bibinfo{person}{Ke Yi}.}
  \bibinfo{year}{2022}\natexlab{}.
\newblock \showarticletitle{Conjunctive Queries with Comparisons}. In
  \bibinfo{booktitle}{\emph{Proceedings of the 2022 International Conference on
  Management of Data}}. \bibinfo{pages}{108--121}.
\newblock


\bibitem[Xiao et~al\mbox{.}(2022)]%
        {xiao2022fast}
\bibfield{author}{\bibinfo{person}{Renjie Xiao}, \bibinfo{person}{Zijing Tan},
  \bibinfo{person}{Haojin Wang}, {and} \bibinfo{person}{Shuai Ma}.}
  \bibinfo{year}{2022}\natexlab{}.
\newblock \showarticletitle{Fast approximate denial constraint discovery}.
\newblock \bibinfo{journal}{\emph{Proceedings of the VLDB Endowment}}
  \bibinfo{volume}{16}, \bibinfo{number}{2} (\bibinfo{year}{2022}),
  \bibinfo{pages}{269--281}.
\newblock


\bibitem[Zilberstein(1996)]%
        {zilberstein1996using}
\bibfield{author}{\bibinfo{person}{Shlomo Zilberstein}.}
  \bibinfo{year}{1996}\natexlab{}.
\newblock \showarticletitle{Using anytime algorithms in intelligent systems}.
\newblock \bibinfo{journal}{\emph{AI magazine}} \bibinfo{volume}{17},
  \bibinfo{number}{3} (\bibinfo{year}{1996}), \bibinfo{pages}{73--73}.
\newblock


\end{thebibliography}
\end{document}